\documentclass[10pt,twocolumn,twoside]{IEEEtran}

\usepackage{multicol, multirow}
\usepackage{array}
\usepackage{amsmath, amsthm, amsfonts, amssymb}

\usepackage{cite}
\usepackage{graphicx}
\usepackage{subfigure}
\usepackage{color}
\usepackage{cases}
\usepackage[]{footmisc}
\usepackage{calc}
\usepackage{physics}
\usepackage{booktabs}
\usepackage{placeins}
\usepackage{xcolor}
\usepackage{hhline}
\usepackage{kotex}
\usepackage{soul}
\usepackage{dblfloatfix}

\newcommand{\Conv}{%
  \mathop{\scalebox{1.5}{\raisebox{-0.2ex}{$\circledast$}}
  }
}

\ifCLASSOPTIONonecolumn
   \renewcommand{\baselinestretch}{2.0}  
   \newcommand{\figwidth}{0.6\columnwidth}
   
   \newcommand{\widefigwidth}{0.5\columnwidth} 
 \else
   \newcommand{\figwidth}{0.9\columnwidth}  
   
   \newcommand{\widefigwidth}{0.95\columnwidth}   
\fi

\newcommand{\ignoreIt}[1] {}

\allowdisplaybreaks 
\newtheorem{theorem}{Theorem}
\newtheorem{corollary}{Corollary}
\newtheorem{lemma}[theorem]{Lemma}
\newtheorem{prop}{Proposition}
\newtheorem{remark}{Remark}

\DeclareMathOperator*{\argmax}{argmax}

\begin{document}
\author{Jinyoung Lee, Duc Trung Dinh, Hyeonsik Yeom, Si-Hyeon Lee, and Jeongseok Ha}


\date{\empty}

\title{Multi-User Cooperation for Covert Communication Under Quasi-Static Fading}
\ifCLASSOPTIONonecolumn
  \renewcommand{\baselinestretch}{1.5}  
\fi

\maketitle

\begin{abstract}
This work studies a covert communication scheme for an uplink multi-user scenario in which some users are opportunistically selected to help a covert user. In particular, the selected users emit interfering signals via an orthogonal resource dedicated to the covert user together with signals for their own communications using orthogonal resources allocated to the selected users, which helps the covert user hide the presence of the covert communication. For the covert communication scheme, we carry out extensive analysis and find system parameters in closed forms. The analytic derivation for the system parameters allow one to find the optimal combination of system parameters by performing a simple one-dimensional search. In addition, the analytic results elucidate relations among the system parameters. In particular, it will be proved that the optimal strategy for the non-covert users is an on-off scheme with equal transmit power.  The theoretical results derived in this work are confirmed by comparing them with numerical results obtained with exhaustive searches. Finally, we demonstrate that the results of work can be utilized in versatile ways by demonstrating a design of covert communication with energy efficiency into account.
\end{abstract}    


\ifCLASSOPTIONonecolumn
  \clearpage
  \pagenumbering{arabic}
\fi

\IEEEpeerreviewmaketitle
\section{Introduction} 
\IEEEPARstart{I}{n} the era of the Internet of Things, wireless communication is at the risk of adversarial eavesdropping with the extremely large amount of personal data in the open wireless networks. To protect wireless privacy, many studies on physical layer security (PLS) have been extensively conducted during the past decades \cite{Bloch11Physical-Layer}, which takes advantage of physical layer resources to resolve security issues such as eavesdropping, key sharing, and low probability of detection (LPD) communication, a.k.a, \emph{covert} communication. Covert communication aims to hide the existence of communication behaviors from a warden, i.e., guaranteeing a very low probability of being detected by the warden, while attaining a certain data rate, called \emph{covert rate}, at a target receiver. Since the early 20th century, spread spectrum techniques \cite{Simon94Spread} have been developed to realize the covert communication for military purposes. While there had been many progresses of spread spectrum techniques, the fundamental limit on the covert rate was still unanswered. Recently, a great deal of efforts \cite{Bash13Limits, Che13Reliable, Bloch16Covert, Wang16Fundamental, Arumugam16Keyless,Tan19Time,Cho21Treating, Lee15Achieving,Shahzad17Covert, He18Covert, Wang19Covert, Sobers17Covert, Shahzad18Achieving, Hu19Covert, Li20Optimal, Zheng21Wireless} have been paid to answering the question. 

In \cite{Bash13Limits}, the authors uncover the square-root law of covert rate which assumes covert communication on additive Gaussian white noise (AWGN) channels. In particular, the square-root law states that at most $\mathcal O (\sqrt n)$ bits can be delivered reliably and covertly to the intended receiver in $n$ channel uses. However, the square-root law unfortunately implies a pessimistic conclusion that the achievable covert rate per channel use,  $\mathcal O (\sqrt n/n)$ approaches 0 as $n \rightarrow \infty$. Later, the information-theoretic limit of covert communication has been further investigated for various channels such as binary symmetric channels \cite{Che13Reliable}, discrete memoryless channels \cite{Bloch16Covert,Wang16Fundamental}, multiple access channels \cite{Arumugam16Keyless}, broadcast channels \cite{Tan19Time}, and interference channels \cite{Cho21Treating}. The studies in \cite{Che13Reliable, Bloch16Covert, Wang16Fundamental, Arumugam16Keyless, Tan19Time, Cho21Treating} again reached the same pessimistic square-root law, which serves as a stimulus for subsequent researches.

Since the fundamental studies on the limit of covert communication \cite{Bash13Limits, Che13Reliable, Bloch16Covert, Wang16Fundamental, Arumugam16Keyless,Tan19Time, Cho21Treating}, the covert communication has been further investigated with various practical considerations such as uncertainty of channel parameters \cite{Lee15Achieving, Shahzad17Covert}, jammers \cite{Sobers17Covert, He18Covert, Li20Optimal, Zheng21Wireless}, relays \cite{Wang19Covert}, full-duplex receiver \cite{Shahzad18Achieving, Hu19Covert}. The authors in \cite{Lee15Achieving} demonstrate that a positive covert rate is achievable on AWGN channels when the warden has uncertainty about the variance of the background noise. They find that there exists a certain signal-to-noise ratio (SNR) called \emph{SNR wall} below which a positive covert rate is achievable. Meanwhile, Shahzad \emph{et. al} in \cite{Shahzad17Covert} also investigate covert communication in the presence of channel uncertainty. In particular, it is assumed that the warden has noisy channel estimations for the channels between the transmitter and receivers. The authors show that the uncertainty of channel information can be leveraged to increase the covert rate. In addition, Wang \emph{et. al} in \cite{Wang19Covert} investigate covert communication in relay networks where the warden suffers from channel estimation errors. For the setup, the detection error probability (DEP) of the warden and covert rate are analyzed, which shows that the transmitter can send $\mathcal O(n)$ bits reliably and covertly in $n$ channel uses.

For breaking the square-root law, Sobers \emph{et. al} in \cite{Sobers17Covert} use a single jammer which deliberately induces uncertainty  to the warden. They prove that $\mathcal O(n)$ bits can be delivered reliably and covertly to the intended receiver in $n$ channel uses when the channel between the jammer and the warden is either an AWGN channel or a block fading channel. Furthermore, He \emph{et. al} \cite{He18Covert} consider a covert communication scenario with multiple jammers and analyze the covert rate with the aid of stochastic geometry. In \cite{Shahzad18Achieving, Hu19Covert}, a full-duplex receiver is employed and plays the role of jammer while receiving covert messages. In particular, the work in \cite{Shahzad18Achieving} assumes that the full channel state information (full CSI) of the channel between the transmitter and the warden is revealed to the warden. The authors in \cite{Shahzad18Achieving} derive a closed-form expression of minimum DEP and jointly optimize the distributions of the jamming power and the transmission probability. Meanwhile, Hu \emph{et. al} in \cite{Hu19Covert} consider the case that the warden is aware of only the distribution of the channel between the transmitter and the warden and find the optimal transmit power for maximizing covert rate with the channel inversion power control (CIPC) and truncated CIPC. The work in \cite{Li20Optimal} prove that the optimality of CIPC when the transmitter knows the full CSI of the channel between the transmitter and the receiver, and the warden a priori knows only the type of channel between the jammer and the warden, i.e., either an AWGN channel or a Rayleigh channel. However, the works in \cite{Hu19Covert, Li20Optimal} find key design parameters via numerical evaluations, which makes it difficult to discover relations among the parameters. In addition, the models in \cite{Hu19Covert, Li20Optimal} are limited to the single jammer case. Later, Zheng \emph{et. al} \cite{Zheng21Wireless} consider a cooperative jamming scheme which selects jammers based on the magnitudes of channel gains between the jammers and the receiver. In particular, jammers are selected when their channel-gain-magnitudes are less than a certain activation threshold in order to minimize the interference to the receiver. In the cooperative jamming scheme, the activation threshold is designed to maximize covert rate. However, most of the results in \cite{Zheng21Wireless} are obtained via numerical evaluations, which fails to provide insights between the design parameters and the covert rate. In addition, the cooperative jammers are assumed to have an equal transmit power whose optimality however is not discussed. 

This work considers the uplink transmission of a multi-user system where all users access a base station (BS) called, Bob through orthogonal resources, e.g., frequency bands, time slots, and spreading sequences. Among the users, a covert user, called Alice transmits her covert messages with a certain probability, which is overheard by a warden, called Willie. To hide the transmission from Alice, we propose a scheme in which some selected users, called cooperative users, transmit interference signals through the orthogonal resource allocated to Alice while they transmit their own signal through their designated orthogonal resources. The user selection is conducted by utilizing the multi-user diversity in a careful way so that the interference at Bob is minimized while the detection capability of covert communication at Willie is most hampered. Although this work looks somewhat similar to the existing one \cite{Zheng21Wireless}, the unique contributions of this work are summarized as follows: 
\begin{enumerate}[\topsep = 0pt]
  \item This work proposes a covert communication scheme for an uplink multi-user scenario in which some users are opportunistically selected to help a covert user, Alice, while transmitting their own messages. 
    
  \item We carry out extensive analysis for the proposed system, which allows us to find the optimal power profile of the interference signals from the selected users called \emph{cooperative} users. The analysis in this work reveals that an on-off strategy with equal transmit power is optimum, and thus, the joint optimization of the user selection and power profile turns out to be a simple on-off scheme. 
    
\item The analysis also enables us to find the optimal detection threshold of Willie in a closed-form which minimizes DEP. The optimal detection threshold in turn allows us to derive the minimum number of cooperative users to meet a target DEP in an explicit form.

\item The system consists of multiple inter-related parameters of which optimization seems complicated. However, it is shown that the transmit power of the covert user, i.e., Alice, can be obtained by performing a simple one-dimensional search. The transmit power of Alice from the search allows us to analytically obtain all the other system parameters. 
\end{enumerate}

The rest of this paper is organized as follows. In Section \ref{Sec:System}, the system model and the problem addressed in this work will be introduced. In Section \ref{Sec:PowerProfile}, we will investigate the optimal power profile for cooperative users and reformulate the optimization problem. The optimal detection threshold, the minimum number of cooperative users, and the activation threshold will be derived in closed-form expressions in Section \ref{Sec:DEP_Willie}. Based on the results in Section \ref{Sec:DEP_Willie}, we will derive the parameters governing the throughput of covert communication, i.e., the connection probability and the maximum covert rate in Section \ref{Sec:Opt_CT}. The theoretical results of this work will be evaluated and compared with numerical results in Section \ref{Sec:numerical} where some insights from the comparisons will also be discussed. Finally, in Section \ref{Sec:Conclusions}, we will conclude this work.

\section{System Model and Problem Formulation} \label{Sec:System}

\subsection{System Model} 
We consider the uplink of a multi-user system, in which $M$ non-covert users and a covert user access a BS, i.e., Bob, through dedicated orthogonal resources. This work assumes the frequency division multiple access (FDMA) without loss of generality. In Fig. \ref{fig:System Model}, the covert user and each of $M$ non-covert users are denoted by Alice and $U_m$ for $m = 1, 2, \ldots, M$, respectively, and some of the non-covert users are selected as cooperative users. We will simply refer the collection of covert and non-covert users to as `users' when there is no need to distinguish them. It is assumed that all the entities in Fig. \ref{fig:System Model} are equipped with single antennas, and each user transmits its message, $W_m$ over a dedicated frequency band, $f_m$ while Alice transmits her message $W_a$ over the frequency band, $f_a$ with a certain probability. 

The channels between the users and Bob are estimated by performing a two-way channel estimation in which Bob first broadcasts a pilot signal over all the frequency bands, i.e., $f_a$ and $f_m$ for $1 \le m \le M$. Then, all the users including the covert user, i.e., Alice, receive the pilot signal which is transmitted over the corresponding frequency bands. Based on the received pilot signals, user $m$ and Alice estimate the channel gains $h^m_{b, m}$ and $h^a_{b, a}$, respectively, where the symbols, $m$, $a$, and $b$ in the subscript and superscript indicate the user index, Alice, and Bob, respectively. The symbol in the superscript tells the allocated frequency band. The pair of symbols in the subscript indicates the nodes for which the channel gain is obtained with the pilot signal transmitted from the node designated by the first symbol to the one designated by the second symbol. That is, $h^m_{b, m}$ is the channel gain between user $m$ and Bob with the pilot signal transmitted from Bob to user $m$ over the frequency band for user $m$, i.e., $f_m$. This work assumes that the frequency bands $f_a$ and $f_m$ for $m = 1, 2, \ldots, M$ are separated larger than the coherent bandwidth so that the channel gains in $\mathcal{H}_d = \{h^a_{b, a}, h^m_{b, m}: 1 \le m \le M\}$ are statistically independent.

Then, in the second phase of the two-way channel estimation, each non-covert user, $U_m$ transmits a \emph{secret} orthogonal pilot signal \cite{Bash13Limits, Zheng21Wireless, Xu19Pilot} back to Bob over dual frequency bands $f_a$ and $f_m$ while Alice sends her secret pilot signal back to Bob over her dedicated frequency band $f_a$. The secret pilot signals are a priori shared between the users and Bob, and thus enable Bob to estimate the channel gains $h^a_{m, b}$ and $h^a_{a, b}$ for $m = 1, 2, \ldots, M$. This work assumes that the transmissions of uplink and downlink pilot signals are performed in the same coherent time, and thus the channel reciprocity holds, i.e., $h^x_{y, z} = h^x_{z, y}$ for $h^x_{y, z} \in \mathcal H_u \cup \mathcal H_d$ where $\mathcal H_u = \{h^a_{a, b}, h^a_{m, b}, h^m_{m, b}: 1 \le m \le M\}$. It is also assumed that the channel estimations are conducted without estimation error, and the channel gains in $\mathcal H_d \cup \mathcal H_u$ have the same statistical properties. Note that since the downlink pilot signals are not secured, the attacker, i.e., Willie, can listen to the downlink pilot signals, which allows Willie to estimate the channels $h^m_{b, w}$ and $h^a_{b, w}$. However, the transmissions of secret uplink pilot signals make Willie ignorant of the channel gains. Thus, Willie has the knowledge of only the statistical properties for $h^a_{m, w}$ and $h^a_{a, w}$. The work assumes quasi-static Rayleigh fading channels, i.e., the channel gains remain constants within one codeword and change independently for another codeword. Thus, the channel gains follow circularly symmetric complex Gaussian distributions with zero mean and variances $\lambda_b$ and $\lambda_w$, respectively, i.e., $h^m_{m, b}$, $h^a_{a, b}$ and $h^a_{m, b} \sim \mathcal{CN}(0,\lambda_b)$, and $h^a_{a, w}$, $h^m_{m, w} $ and $h^a_{m, w} \sim \mathcal{CN}(0,\lambda_w)$, where the variances $\lambda_b$, and $\lambda_w$, are assumed to be unity, i.e., $\lambda_b = \lambda_w = 1$, which implies that Bob and Willie are located at the same distance from the users.

\begin{figure}[t]
\centering
\includegraphics[width=\widefigwidth]{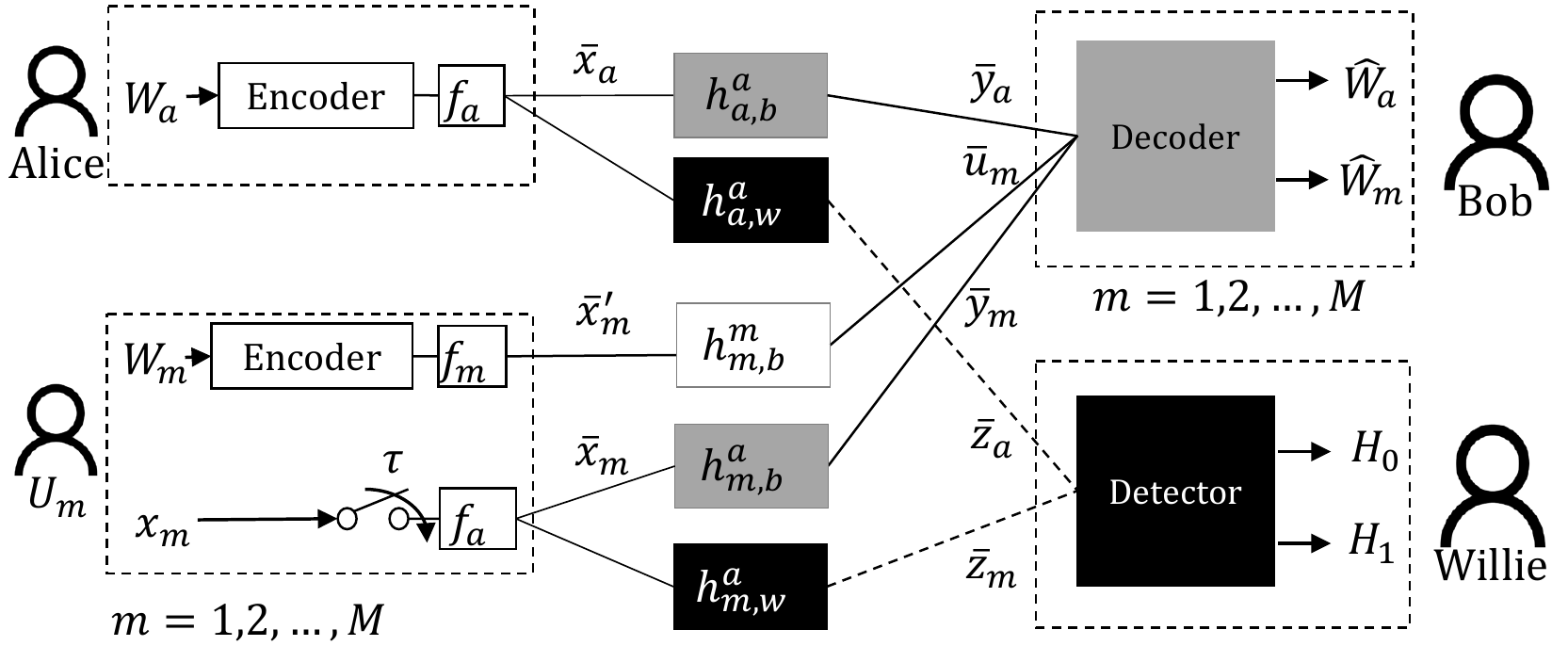}
\caption{System model of multi-user cooperation for covert communication.}
\label{fig:System Model} 
\end{figure}

After performing the channel estimations, each non-covert user transmits its signal, $\bar x'_m = \{x'_m[t]: 1 \le t \le N\}$ in Fig. \ref{fig:System Model}, for the message $W_m$ over the dedicated frequency band, $f_m$ where $N$ is the codeword length. In the meantime, each cooperative user emits its interference signal, $\bar x_m = \{x_m[t]: 1 \le t \le N\}$ in Fig. \ref{fig:System Model}, over the frequency band $f_a$ together with the signal for its own message, i.e., $\bar x'_m$ over $f_m$. As mentioned earlier, the interference signals help hide the transmission of the covert message but induce undesirable interference in the reception of the covert message at Bob. Thus, the selection of the cooperative users must be conducted in a careful way, which will be introduced shortly. Meanwhile, the covert user, Alice transmits her signal, $\bar x_a = \{x_a[t]: 1 \le t \le N\}$ in Fig. \ref{fig:System Model}, for the message $W_a$ with a certain probability.  Since we are interested only in the transmission of the covert message, we focus only on the received signals for the mixture of the transmitted signals from the covert and cooperative users over the frequency band $f_a$.

Through the frequency band $f_a$, the received signal at Bob is expressed as
\begin{align} \label{eq:y_Bob}
    y[t]= 
  \begin{cases}
    \displaystyle\sum_{m = 1}^{M} \sqrt{P_m} h^a_{m, b} x_m[t] + n_b[t], & H_0 \\
    \displaystyle\sqrt{P_a}h^a_{a, b} x_a[t] + \sum_{m = 1}^{M} \sqrt{P_m}h^a_{m, b}x_m[t] + n_b[t], & H_1
  \end{cases},
\end{align}
where $x_a[t]$ and $x_m[t]$ are the transmitted signals from Alice and user $m$, respectively and follow the zero-mean Gaussian distribution with the unit variance, i.e., $\mathbb E[|x_a[t]|^2]= \mathbb E[|x_m[t]|^2]=1$ for $t = 1, 2, \ldots, N$, and $P_a$ and $P_m$ denote the transmit power of Alice and the interference signal power from user $m$, respectively. Note that cooperative users can be identified as the ones with nonzero interference power, i.e., $P_m > 0$. In \eqref{eq:y_Bob}, the signal, $n_b[t]$ represents the AWGN at Bob with variance $\sigma_b^2$, i.e., $n_b[t] \sim \mathcal{CN}(0,\sigma_b^2)$, and $H_0$ and $H_1$ denote the hypotheses that Alice transmits her covert message or not, respectively. Meanwhile, the received signal at Willie is given by
\begin{equation} \label{eq:y_Willie}
    z[t]= 
  \begin{cases}
      \displaystyle\sum_{m = 1}^{M} \sqrt{P_m}h^a_{m, w}x_m[t] + n_w[t], & H_0 \\
      \sqrt{P_a}h^a_{a, w}x_a[t] \\ ~~~~+ \displaystyle\sum_{m = 1}^{M}\sqrt{P_m}h^a_{m, w}x_m[t] + n_w[t], & H_1
  \end{cases},
\end{equation} 
where $n_w[t]$ is the AWGN at Willie with variance $\sigma_w^2$, i.e., $n_w[t] \sim \mathcal{CN}(0,\sigma_w^2)$.

\subsection{Detection Error at Willie}\label{SubSec:metric_Willie}
Based on the observation at Willie in \eqref{eq:y_Willie}, he attempts to determine whether Alice transmits her covert message or not. The ultimate goal of Willie is to find an optimal strategy to detect the covert communication between Alice and Bob based on his received signal. The detection metric of Willie is defined as the error probability as follows:
\[
  P_e = \Pr(H_0)P_{\text FA} + \Pr(H_1)P_{\text MD},
\]
where $\Pr(H_1)$ and $\Pr(H_0)$ are the probabilities for Alice to transmit her covert message and to be silent, respectively, $P_{\text FA} = \Pr(D_1 | H_0)$ is the false alarm probability, $D_1$ indicates the event for Willie to decide the presence of  covert communication, and $P_{\text MD} = \Pr(D_0|H_1)$ is the miss-detection probability, and $D_0$ indicates the event for Willie to decide no presence of covert communication. According to \cite{Sobers17Covert, Shahzad18Achieving, Hu19Covert}, $P_e \geq \min\{\Pr(H_0),\Pr(H_1)\} (P_{\text FA} + P_{\text MD})$. Hence, we define the covert constraint as $P_{\text FA} + P_{\text MD} \geq 1-\epsilon$ for arbitrary $\epsilon > 0$. For simplicity, we denote the DEP by $\zeta$, i.e., $\zeta = P_{\text FA} + P_{\text MD}$. Note that a blind test with no side information satisfies $P_{\text FA} + P_{\text MD} = 1$. This work assumes that Willie employs an energy detector \cite{Zheng21Wireless, Hu19Covert, Li20Optimal} to detect the covert communication since it was shown in \cite{Shahzad21Covert} that the energy detector is asymptotically optimal in the large DEP regime.


The test statistic for the energy detector is given by
\begin{equation} \label{Eq:TestStat}
  T_w = \frac{1}{N} \sum_{t=1}^{N} |z[t]|^2. 
\end{equation}
Using the test statistic, the decision is made with the following criterion: 
\begin{equation} \label{eq:test}
   T_w \mathop{\lessgtr}_{D_1}^{D_0} \gamma,
\end{equation}
where $\gamma$ is a predetermined detection threshold. Considering that $x_m[t]$ and $n_w[t]$ are uncorrelated Gaussian random variables, the test statistic in \eqref{Eq:TestStat} can be expressed as 
\[
  T_w= 
  \begin{cases}
    \displaystyle\left(\sum_{m=1}^{M} P_m|h^a_{m,w}|^2 + \sigma_w^2 \right) \cdot \frac{\chi^2_{2N}}{2N}, & H_0 \\
    \displaystyle\left(\sum_{m=1}^{M} P_m|h^a_{m,w}|^2 +P_a|h^a_{a,w}|^2 + \sigma_w^2\right) \cdot \frac{\chi^2_{2N}}{2N}, & H_1
  \end{cases},
\]
where $\chi^2_N$ is a chi-square random variable with $N$ degrees of freedom, and Willie knows only the probability distributions of the channel gains, i.e., $h^a_{m,w}$ and $h^a_{a,w}$. It is assumed that the length of codeword is sufficiently long, i.e., $N \gg 1$, and Willie can collect a large number of samples. In this setup, due to the strong law of large numbers, $\chi^2_{2N}/2N$ converges to 1 when $N \rightarrow \infty$, and thus, the test statistic can be reformulated as
\begin{align}\label{eq:test_statistic}
   T_w= 
  \begin{cases}
    \displaystyle\sum_{m=1}^{M} P_m|h^a_{m,w}|^2 + \sigma_w^2, & H_0 \\
    \displaystyle\sum_{m=1}^{M} P_m|h^a_{m,w}|^2 +P_a|h^a_{a,w}|^2 + \sigma_w^2 , & H_1
  \end{cases}.
\end{align}
Then, DEP, $\zeta$, is given by
\begin{equation} \label{eq:DEP}
  \zeta  = \Pr(T_w > \gamma ; H_0) + \Pr(T_w \leq \gamma ; H_1).  
\end{equation}
Note that the detection error probability, $\zeta$ is a function of the detection threshold $\gamma$, and we define the minimum of $\zeta$ as 
\[
  \zeta_{\rm min} = \min_\gamma \zeta.
\]

\subsection{Problem Formulation}\label{SubSec:Problem}

Since the quasi-static fading is considered in this work, an outage probability, denoted by $P_o$, can be defined as 
\begin{equation}\label{eq:connection}
P_o = \Pr\left(\log_2\left(1+\frac{P_a|h^a_{a,b}|^2}{\sum_{m=1}^M  P_m |h^a_{m,b}|^2 + \sigma_b^2}\right) < R\right),
\end{equation}
where $R$ is a target covert rate. The outage probability enables us to define the probability of the complementary event, i.e., the connection probability, as $P_c = 1 - P_o$. As a performance measure, we consider the throughput of the covert communication which is given by 
\begin{equation}
  \eta = RP_c.
\end{equation}
The throughput has been widely adopted as a performance measure in the existing works \cite{Zheng21Wireless, Hu19Covert, Li20Optimal} on covert communication over quasi-static fading channels. The design of covert communication system introduced in this work can be formulated into an optimization problem as follows:
\begin{subequations}\label{eq:original_opt}
\begin{align}
  &(R^*, P_a^*, \bar{P}^*) = \argmax_{R, P_a, \bar P} ~ \eta\label{eq:Opt_obj} \\
  & ~~~~ \text{s.t.} ~~ \zeta_{\rm min} \geq 1- \epsilon \label{eq:Opt_covertness}\\
  & ~~~~~~~~~ 0 \leq P_a \leq P_{\max}\\
  & ~~~~~~~~~ 0\leq P_m \leq P_{\max}, \label{eq:Opt_power}
\end{align}
\end{subequations}
where $\bar P \triangleq (P_1,P_2,\ldots,P_M) $ is the transmit power profile for users, and $P_{\max}$ is the maximum transmit power of user $m$ for $1 \le m \le M$. Note that cooperative users are the ones with non-zero interference power, i.e., $U_m$ for $m \in \{n : P_n > 0, 1 \le n \le M\}$.

\section{Optimal Power Profile for Cooperative Users}\label{Sec:PowerProfile}
For solving the optimization problem in \eqref{eq:original_opt}, we first find the optimal power profile for given $R$ and $P_a$.
\begin{prop}\label{prop:power_profile}
The on-off scheme with the maximum transmit power is the optimal power allocation of the users for given $R$ and $P_a$. That is,
\[
       P_m^*= 
  \begin{cases}
      P_{\max}, & |h^a_{m,b}|^2 \leq \tau \\
      0 , & {\text otherwise}
  \end{cases}.
\]
The activation threshold, $\tau$ is decided to make the inequality, $\zeta_{\min} \ge 1 - \epsilon$ hold while minimizing $\Omega = \sum_m P^2_m$. 
\end{prop}
\begin{proof}
The proof is provided in Appendix \ref{Appendix:power_profile}.
\end{proof}

\begin{remark}\label{remark:Determination_K}
Bob broadcasts $\tau$ which is determined by $h^a_{m,b}$ after the two-way channel estimation. The details of deciding $\tau$ will be shortly given in Corollary \ref{corollary:tau}. Note that $\Omega = \sum_m P^2_m$ can be expressed as $K P^2_{\max}$ where $K$ is the number of users whose channel gains satisfy the inequality, $|h^a_{m,b}|^2 \leq \tau$. The parameter $\tau$ is decided to make the inequality, $\zeta_{\min} \ge 1 - \epsilon$ hold while minimizing $\Omega$, which is equivalent to minimizing $K$ while satisfying the covert constraint, $\zeta_{\min} \ge 1 - \epsilon$. It can be noticed that Bob selects the $K$ smallest channels based on the knowledge of $h^a_{m,b}$. However, Willie does not know who are selected to emit the interference signals among the users despite of knowing $\tau$ since $\tau$ is determined by only $h^a_{m,b}$'s which are however unknown to Willie.
\end{remark}

This work finds the optimal power profile of cooperative users while the authors in \cite{Zheng21Wireless} just assume a given power for the cooperative jammers. In \cite{Li20Optimal,Hu19Covert}, the power of single jammer is also assumed to be fixed to a certain value. Thus, this work finds for the first time the power profile of multiple cooperative users for covert communication. Based on the optimal power profile, we can reformulate the original optimization problem (\ref{eq:original_opt}) to 
\begin{subequations}\label{eq:reformul_opt}
\begin{align}
  &\left(R^*, P_a^*, \tau^* \right)=\argmax_{R, P_a, \tau} ~ \eta  \\
  & ~~~~ \text{s.t.} ~~ \zeta_{\rm min} \geq 1- \epsilon \label{eq:CvtConst}\\
  & ~~~~~~~~~ 0 \leq P_a \leq P_{\max} \label{eq:PwrConst}.
  \end{align}
\end{subequations}
From Appendix \ref{Appendix:power_profile}, it can be noticed that $\zeta_{\min}$ is a monotonically increasing function of the square sum of the interference power, i.e., $\zeta_{\min} = f(KP_{\max}^2)$ where $f(\cdot)$ is a monotonically increasing function of its argument. We will discuss how to find $K$ and $\tau$ in the next section.

\section{Detection Performance for Willie} \label{Sec:DEP_Willie}
In this section, we first derive the detection error probability, $\zeta$ in \eqref{eq:DEP}. Then, we find out the detection threshold $\gamma^\ast$ at which the minimum detection error probability, $\zeta_{\min}$ is obtained. The number of cooperative users, $K$, and the activation threshold, $\tau$ are also derived in closed forms.

The test statistic in \eqref{eq:test_statistic} can be expressed as
\begin{align}\label{eq:test_statistic2}
   T_w= 
  \begin{cases}
    \displaystyle\sum_{i=1}^{K} P_{\max}|h^a_{m_i, w}|^2 + \sigma_w^2, & H_0 \\
    \displaystyle\sum_{i=1}^{K} P_{\max}|h^a_{m_i, w}|^2 +P_a|h^a_{a,w}|^2 + \sigma_w^2 , & H_1
  \end{cases},
\end{align}
where $m_i$'s for $1 \le i \le K$ are the indices of the selected users, i.e., the cooperative users in Proposition \ref{prop:power_profile}. Then, the test statistic $T_w$ in \eqref{eq:test_statistic2} follows the following probability distributions depending on the hypothesis:
\begin{align}\label{eq:test_statistic_RV}
  T_w - \sigma_w^2  \sim
   \begin{cases}
      \Gamma(K,P_{\max}), & H_0, \\
      \Gamma(K,P_{\max}) \circledast  \Gamma(1,P_a), & H_1,
  \end{cases}
\end{align}
where $\Gamma(a,b)$ is the Gamma distribution with shape parameter $a$ and scale parameter $b$, and $\circledast$ indicates the convolution between the two Gamma distributions. Note that the sum of squares of channel-gain magnitudes follows the Gamma distribution. Based on the distributions of the test statistic in \eqref{eq:test_statistic_RV}, we can have a closed-form expression of the detection error probability in Lemma \ref{lemma:DEP}.
\begin{lemma}\label{lemma:DEP}
  The detection error probability, $\zeta$, is given by  
\begin{align}\label{eq:exact_DEP}
      \zeta &= 1 - \exp\left(K\frac{P_{\max}^2+2P_aP_{\max}}{2P_a^2} - \frac{\bar{\gamma}}{P_a}\right) \nonumber\\
    & \hspace{0.2\columnwidth} \times Q\left(\sqrt{K}\frac{P_{\max}+P_a}{P_a} - \frac{1}{P_{\max}\sqrt{K}}\bar{\gamma}\right),
\end{align}
where $\bar\gamma \triangleq \gamma - \sigma_w^2$.
\end{lemma}
\begin{proof}
The proof is provided in Appendix \ref{Appendix:DEP}.
\end{proof}
The detection error probability derived in this work enables us to find out some of the system parameters in closed forms. In particular, utilizing the detection error probability in \eqref{eq:exact_DEP}, we will derive the optimal detection threshold in a closed form, $\gamma^\ast$ at which the detection error probability in \eqref{eq:exact_DEP} is minimized, i.e.,
\[
  \gamma^* = \arg \min_{\gamma} \zeta.
\]
The optimal detection threshold $\gamma^\ast$ is given by Lemma \ref{lemma:gamma}.
\begin{lemma} \label{lemma:gamma}
  The optimal detection threshold at Willie is given by 
  \begin{equation}\label{eq:Opt_gamma}
     \gamma^* = KP_{\max}+\sigma_w^2.
   \end{equation} 
\end{lemma}
\begin{proof}
From the detection error probability in \eqref{eq:exact_DEP}, we have 
\begin{align*}
  \zeta = 1 - & \exp\left(K\frac{P_{\max}^2 + 2P_aP_{\max}}{2P_a^2} - \frac{\bar{\gamma}}{P_a}\right) \nonumber\\
     & \hspace{0.2\columnwidth} \times Q\left(\sqrt{K}\frac{P_{\max}+P_a}{P_a} - \frac{1}{P_{\max}\sqrt{K}}\bar{\gamma}\right) \\     
    & \mathop{\approx}_{(a)} 1 - \exp(-\frac{\bar{\gamma}^2}{2KP_{\max}^2}+\frac{\bar{\gamma}}{P_{\max}} - \frac{K}{2}),\\
\end{align*}
where the approximation in (a) is from $Q(x) \approx e^{-x^2/2}$, which is validated since $\sqrt{K}(P_{\max}+P_a) \gg P_a$. Then, the minimization of $\zeta$ is equivalent to maximizing the argument of the exponential function which is a quadratic function of $\gamma$. It is readily found that DEP is minimized at $\gamma^* = KP_{\max} + \sigma_w^2$.
\end{proof}

With the optimal detection threshold, $\gamma^*$ in \eqref{eq:Opt_gamma}, we can derive the minimum DEP in Theorem \ref{Theorem:Min_DEP}.
\begin{theorem}\label{Theorem:Min_DEP}
 The minimum detection error probability at Willie, $\zeta_{\min}$ is given by
  \begin{equation}\label{eq:Opt_DEP}
    \zeta_{\min} = 1-\frac{1}{\sqrt{\pi}\left(\sqrt{\frac{KP_{\max}^2}{2P_a^2}} + \sqrt{\frac{KP_{\max}^2}{2P_a^2} + \frac{4}{\pi}} \right)}.
  \end{equation}
\end{theorem}
\begin{proof}
  Substituting $\gamma^*$ obtained in (\ref{eq:Opt_gamma}) into (\ref{eq:exact_DEP}), we have
  \begin{align} 
    \zeta_{\min} &= 1-\exp(\frac{KP_{\max}^2}{2P_a^2})Q\left(-\sqrt{\frac{KP_{\max}^2}{P_a^2}}\right) \nonumber\\
    &  = 1 - \frac{1}{2}\exp(X)\text{erfc}\left(\sqrt X\right) \nonumber \\
    & \mathop{\approx}_{(a)} 1 - \frac{1}{2}\exp(X)\frac{2\exp(-X)}{\sqrt{\pi}\left(\sqrt{X} + \sqrt{X + \frac{4}{\pi}} \right)} \nonumber \\
    & = 1-\frac{1}{\sqrt{\pi}\left(\sqrt{\frac{KP_{\max}^2}{2P_a^2}} + \sqrt{\frac{KP_{\max}^2}{2P_a^2} + \frac{4}{\pi}} \right)} \nonumber,
  \end{align}
where $X \triangleq K/2\cdot(P_{\max}/P_a)^2$, and the approximation in (a) is from \cite{Whittaker90ACourse}, which is validated since $KP_{\max}^2 > 2P_a^2$. 
\end{proof}
As discussed in Remark \ref{remark:Determination_K}, $K$ is set to the minimum number of cooperative users  satisfying $\zeta_{\min} \geq 1-\epsilon$. From the relation between $K$ and $\zeta_{\min}$, we have the minimum number of cooperative users denoted by $K_{\min}$ in Theorem \ref{Theorem:K}.

\begin{theorem}\label{Theorem:K}
 The minimum number of cooperative users, $K_{\min}$, to satisfy the covert constraint, $\zeta_{\min} \geq 1-\epsilon$ is given by
  \begin{equation} \label{eq:Opt_K}
    K_{\min} = \left\lceil \frac{P_a^2 c_{\epsilon}}{P_{\max}^2} \right\rceil,
  \end{equation}
  where $c_{\epsilon} = \frac{\frac{1}{\epsilon^2} - 8 + 16\epsilon^2}{2\pi}$.
\end{theorem}
\begin{proof}
We find the minimum number of cooperative users, $K_{\min}$ which satisfies the covert constraint, $\zeta_{\min} \ge 1-\epsilon$. Remark \ref{remark:Determination_K} tells that the minimum DEP, $\zeta_{\min}$ is an increasing function of $\Omega = K P_{\max}^2$. Thus, the value of $\zeta_{\min}$ satisfying the covert constraint with equality, i.e., $\zeta_{\min} = 1-\epsilon$ provides $K_{\min}$. We have a relation between $K$ and $\zeta_{\min}$ in \eqref{eq:Opt_DEP}. That is,
\begin{equation}
   \zeta_{\min} =  1-\frac{1}{\sqrt{\pi}\left(\sqrt{\frac{K_{\min} P_{\max}^2}{2P_a^2}} + \sqrt{\frac{K_{\min} P_{\max}^2}{2P_a^2} + \frac{4}{\pi}} \right)} = 1 - \epsilon.
\end{equation}
Some algebraic manipulations provide us with the minimum number of cooperative users as follows:  
\begin{equation}
    K_{\min} =  \frac{P_a^2}{P_{\max}^2}\frac{\frac{1}{\epsilon^2}-8+16\epsilon^2}{2\pi}, \label{eq:Opt_p} 
\end{equation}
which however must be an integer value, and thus we finally have
\[
   K_{\min} =  \left\lceil \frac{P_a^2}{P_{\max}^2}\frac{\frac{1}{\epsilon^2}-8+16\epsilon^2}{2\pi} \right\rceil.
\]
\end{proof}
\noindent Theorem \ref{Theorem:K} says that when $K \ge K_{\min}$, the covert constraint, $\zeta_{\min} \geq 1-\epsilon$ is satisfied. It is also noticed that the number of cooperative users is proportional to the transmit power of Alice, $P_a$. In particular, the number of minimum cooperative users quadratically grows with the increasing transmit power for covert communication. The relation between the covert constraint and the minimum number of cooperative users is also clearly manifested by the result of Theorem \ref{Theorem:K}. That is, the number of minimum cooperative users is quadratically proportional to the inverse of $\epsilon$, i.e., $K_{\min} \sim \mathcal{O}(\epsilon^{-2})$. 

For the minimum number of cooperative users, the optimal activation threshold in Proposition \ref{prop:power_profile} can be readily obtained, which is summarized in Corollary \ref{corollary:tau}.
\begin{corollary}\label{corollary:tau}
The optimal activation threshold, $\tau$ in Proposition \ref{prop:power_profile}, is determined as
\begin{equation}\label{eq:opt_tau}
  \tau = |h^a_{m_{K_{\min}}, b}|^2,
\end{equation}
where $|h^a_{m_1, b}|^2 \le |h^a_{m_2, b}|^2 \le \cdots \le |h^a_{m_M,b }|^2$.
\end{corollary}
\begin{proof}
  By setting $\tau$ to the square of the $K$-th weakest channel magnitude, it is clear that we have $K_{\min}$ cooperative users where $K_{\min}$ is given in \eqref{eq:Opt_K}.
\end{proof}
\noindent Note that Bob can also estimate the minimum DEP, $\zeta_{\min}$, since he knows the statistical distribution of $h^a_{m, w}$. The knowledge of $\zeta_{\min}$ enables Bob to be aware of $K_{\min}$ for a given $\epsilon$ value. Then, Bob determines the activation threshold, $\tau$ in Corollary \ref{corollary:tau} based on $h^a_{m, b}$'s acquired via the two-way channel estimation. Meanwhile, Willie does not know which $K_{\min}$ users are selected since he is ignorant of instantaneous values of $h^a_{m, b}$'s even if he listens to the broadcasted activation threshold $\tau$.


\section{Optimization of Legitimate Users' Parameters} \label{Sec:Opt_CT}
In this section, we first derive the connection probability, $P_c$ based on the results of Theorem \ref{Theorem:K}, and then find out the covert rate $R$ which maximizes the throughput $\eta$.

Alice does not know the instantaneous channel gains between the cooperative user and Bob, i.e., $h^a_{m, b}$, and thus she cannot help deciding her covert rate, $R$ based on the channel statistics. For Alice to decide her covert rate maximizing the throughput $\eta$, it is necessary to derive the connection probability in terms of the covert rate $R$. The connection probability, $P_c = 1 - P_o$ can be expressed as
\begin{align}\label{eq:re_connection}
  P_c & = \Pr\left(\log_2\left(1+\frac{P_a|h_{a,b}|^2}{\sum_{i = 1}^{K_{\min}}  P_{\max} |h_{m_i,b}|^2 + \sigma_b^2}\right) \geq R\right) \nonumber\\
  & = \Pr\left(\frac{P_a|h^a_{a,b}|^2}{\sum_{i = 1}^{K_{\min}} P_{\max} |h^a_{m_i,b}|^2 + \sigma_b^2} \geq r \right) \nonumber\\
  & = \Pr\left( \sum_{i = 1}^{K_{\min}} |h^a_{m_i,b}|^2 \leq \frac{1}{r}\frac{P_a |h^a_{a,b}|^2}{P_{\max}} - \frac{\sigma_b^2}{P_{\max}} \right) \nonumber \\
  & = \Pr\left(S_h \leq T \right) ,
\end{align}
where $m_i$'s are the indices of the cooperative users and $r=2^R - 1$. Note that $T$ in \eqref{eq:re_connection} is known to Alice since she acquires $h^a_{a, b}$ from the channel estimation. Meanwhile, $S_h$ in \eqref{eq:re_connection} is a random variable to Alice. By utilizing the order statistics of $|h^a_{m_i, b}|^2$, we can have the statistical properties of $S_h$, which leads to an expression of $P_c$ in Lemma \ref{lemma:connection_prob}.

\begin{lemma}\label{lemma:connection_prob}
  The connection probability $P_c$ can be expressed as
  \begin{equation}\label{eq:connection_prob}
    P_c = Q\left(\frac{\mu P_{\max} + \sigma_b^2 - \frac{P_a|h^a_{a,b}|^2}{r} }{\Xi P_{\max}}\right),
  \end{equation}
  where $\mu = \sum_{i=1}^{K_{\min}} \frac{K_{\min}+1-i}{M+1-i}$ and $\Xi^2 = \sum_{i=1}^{K_{\min}} \left(\frac{K_{\min}+1-i}{M+1-i}\right)^2$.
\end{lemma}
\begin{proof}
The proof is given in Appendix \ref{Appendix:connection_prob}.
\end{proof}

Based on Lemma \ref{lemma:connection_prob}, we find out the maximum covert rate, $R_{\max}$ maximizing the throughput, $\eta = R P_c$ with the covert constraint in \eqref{eq:CvtConst} for a given $P_a$. We are interested in the large DEP regime, equivalently small $\epsilon$, and the maximum covert rate in the large DEP regime can be expressed as in Theorem \ref{theorem:Opt_R}.
\begin{theorem}\label{theorem:Opt_R}
  For a small $\epsilon$ value and a given $P_a$, the maximum covert rate, $R_{\max}$ subject to the covert constraint in \eqref{eq:CvtConst} is given by
  \begin{equation}\label{eq:Opt_R}
    R_{\max} = \log_2\left(1+\frac{\kappa}{\theta + \sqrt{-\psi}}\right),
  \end{equation}
  where $\theta = \frac{\mu P_{\max}+\sigma_b^2}{P_{\max} \Xi}$, $\kappa = \frac{P_a|h^a_{a,b}|^2}{P_{\max}\Xi}$, and $\psi = 2\log\frac{\sqrt{2\pi}}{2\theta}$. 
\end{theorem}
\begin{proof}
  Details of proof are given by Appendix \ref{Appendix:Opt_R}.   
\end{proof}

With the maximum covert rate, $R_{\max}$, we achieve the maximum throughput $\eta_{\max}$ as shown in Corollary \ref{corollary:Max_Throughput}.
\begin{corollary}\label{corollary:Max_Throughput}
  The maximum throughput, $\eta_{\max}$ is given by 
  \begin{equation}\label{eq:Max_Throughput}
    \eta_{\max}=\log_2\left(1+\frac{\kappa}{\theta + \sqrt{-\psi}}\right) Q(-\sqrt{-\psi}).
  \end{equation}
\end{corollary}

\begin{proof}
  From \eqref{eq:connection_prob} and \eqref{eq:Opt_R}, the maximum throughput at a given $P_a$ becomes
  \[
    \eta_{\max} =  \log_2\left(1+\frac{\kappa}{\theta + \sqrt{-\psi}}\right) Q\left(\frac{\mu P_{\max} + \sigma_b^2 - \frac{P_a|h^a_{a,b}|^2}{r_{\max}} }{\Xi P_{\max}}\right),
  \]
  where $r_{\max} = 2^{R_{\max}}-1$. The connection probability, $P_c$ can be further simplified as
   \begin{align*}
       P_c & = Q\left(\frac{\mu P_{\max} + \sigma_b^2 - \frac{P_a|h_{a,b}|^2}{r_{\max}} }{\Xi P_{\max}}\right) \\
       & =  Q\left(\frac{\mu P_{\max} + \sigma_b^2}{\Xi P_{\max}} - \frac{P_a |h^a_{a,b}|^2 \left(\theta + \sqrt{-\psi}\right)}{\kappa \Xi P_{\max}}\right) \\
       & = Q(-\sqrt{-\psi}).
   \end{align*} 
 This completes the proof. 
\end{proof}

 For a given number of users, $M$ and a small $\epsilon$ value, we can achieve the maximum throughput in a closed-form expression, $\eta_{\max}$ as seen in \eqref{eq:Max_Throughput}. However, the result in \eqref{eq:Max_Throughput} depends on the transmit power for the covert communication, i.e., $P_a$. Thus, we have to find the optimal $P_a$, denoted by $P_a^\ast$ at which $\eta_{\max}$ is maximized. That is,
\begin{subequations}\label{eq:Opt_Pa}
  \begin{align}
    & P_a^*  = \argmax_{P_a} \eta_{\max} \label{eq:Opt_Pa_obj}\\
    & {\rm s.t.} ~~ 0 \leq P_a \leq P_{\max}. \label{eq:Opt_Pa_const}
  \end{align}  
\end{subequations}
Note that the existing works \cite{Zheng21Wireless,Shahzad18Achieving} assume that the transmit power of Alice, i.e., $P_a$ is given. In particular, the authors in \cite{Zheng21Wireless} conduct numerical searches to find all the covert parameters such as the optimal detection threshold, $\gamma^\ast$, the maximum covert rate, $R_{\max}$ and activation threshold, $\tau$ for a given $P_a$. Meanwhile, this work optimizes all the system parameters by conducting the one-dimensional search for $P_a$ in \eqref{eq:Opt_Pa} since all the other system parameters are explicitly expressed as functions of $P_a$.

The analytic results developed in this work can be utilized in various ways. As an example, we consider the covert communication problem with energy efficiency into account. In doing so, we first define the energy efficiency, $E_{\rm eff}$ as the ratio of the throughput to the total power consumption which includes the transmit power of Alice, $P_a$, and the total interference power consumed by the cooperative users. Then, the transmit power, $P_a$ is now determined to maximize the energy efficiency, i.e., 
\begin{subequations}\label{eq:Opt_EE_Pa}
 \begin{align}
   & P_a^*  = \argmax_{P_a} E_{\rm eff}  \label{eq:Opt_EE_Pa_obj}\\
     & {\rm s.t.} ~~ 0 \leq P_a \leq P_{\max}, \label{eq:Opt_EE_Pa_const}
  \end{align}  
\end{subequations} 
where
\begin{equation} \label{eq:Eff}
  E_{\rm eff} = \argmax_{P_a} \frac{\eta_{\max}}{K_{\min}P_{\max} + P_a}.
\end{equation}
Note that all the parameters in the definition of the energy efficiency, $E_{\rm eff}$ can be  efficiently evaluated for a given $P_a$. Thus, the maximization in \eqref{eq:Opt_EE_Pa_const} is readily conducted with the one-dimensional search for $P_a$. In Section \ref{Sec:numerical}, the covert communications with/without the consideration of energy efficient will be compared in terms of throughput and total transmit power.

\section{Numerical results} \label{Sec:numerical}
In this section, we carry out numerical evaluations which are compared with our theoretical results to confirm the analyses in this work. In addition, we will uncover the relations among system parameters. In the evaluations, it is assumed that the maximum transmit power, i.e., $P_{\max}$, is set to 1, and the channel between Alice and Bob is normalized as $|h^a_{a,b}|^2 = 1$. 

We first numerically evaluate DEP, $\zeta$, in Fig. \ref{fig:DEP_gamma}, versus the detection threshold, $\gamma$ with different values of $K$, i.e., $K = 15$ and $25$. The numerical evaluation is performed by generating realizations of the random variable in \eqref{eq:test_statistic_RV}, and an empirical DEP for a detection threshold, $\gamma$, is obtained by conducting the test in \eqref{eq:test} with the realizations. In Fig. \ref{fig:DEP_gamma}, it is witnessed that the minimum DEP values, $\zeta_{\min}$, are obtained at detection thresholds, $\gamma^\ast$ of 0.901 and 0.922 for $K = 15$ and 25, respectively. The detection thresholds for $\zeta_{\min}$ are also evaluated with the theoretical result in \eqref{eq:Opt_gamma}. The comparisons between the numerical and theoretical results look well matched, which confirms the theoretical derivation for $\gamma^\ast$ in \eqref{eq:Opt_gamma}. It should also be noted that the minimum DEP becomes larger with the more number of cooperative users, i.e., $K = 25$, which results in a higher interference power and ends up making Willie more confused. 

\begin{figure}[htb]
\centering
\includegraphics[width=\figwidth]{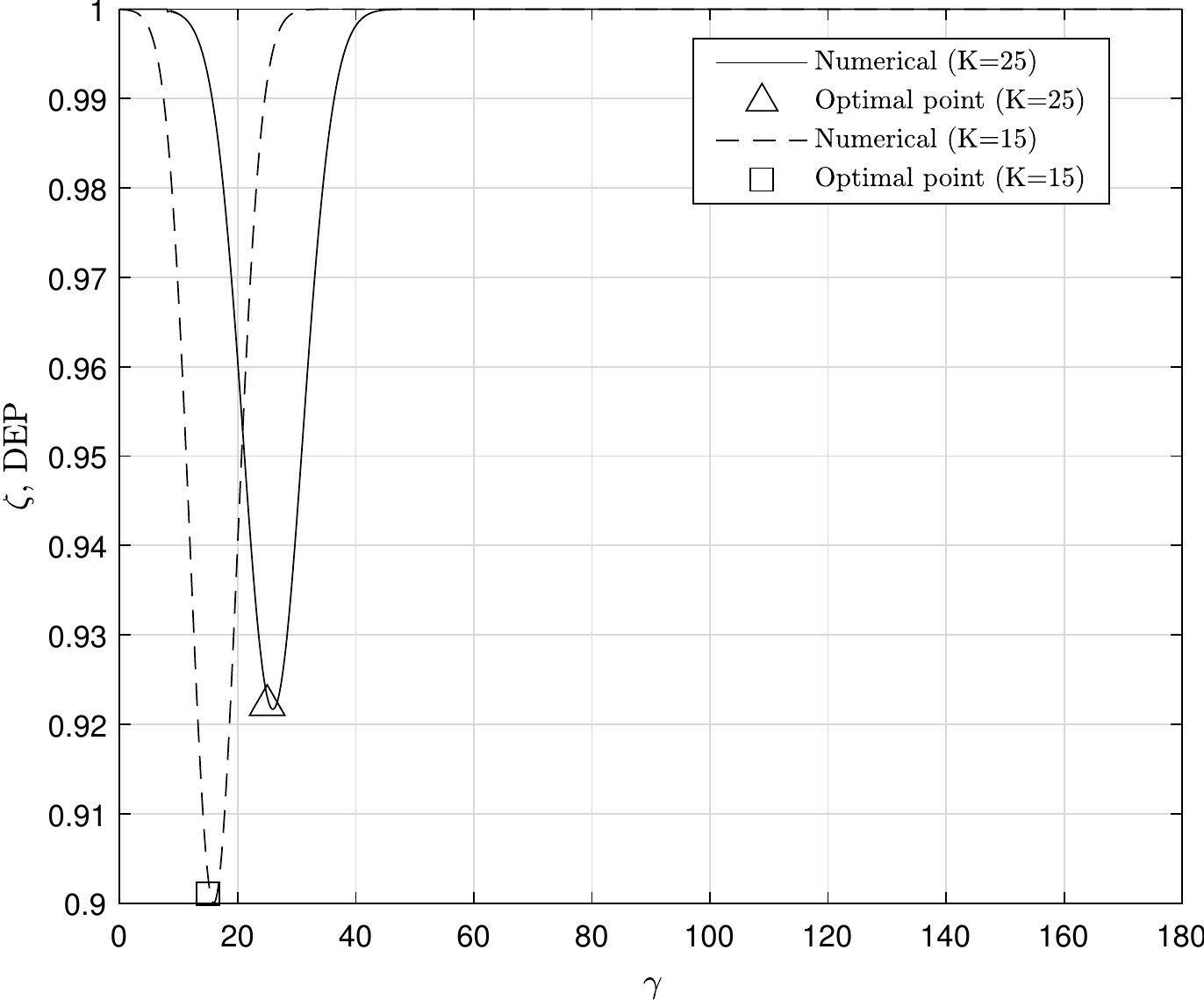}
\caption{Detection error probability, $\zeta$, versus the detection threshold, $\gamma$, with $K = 25$ and $K = 15$, and $P_a = P_{\max}$.}
\label{fig:DEP_gamma} 
\end{figure}

Now, we numerically evaluate the minimum detection error probabilities, $\zeta_{\min}$ for different numbers of cooperative users, $K$ between 1 and 100 setting $P_a$ to either 0.67 or 0.83. The numerical evaluations are obtained by searching for the minimums of the empirical DEP for $K$ values. The evaluations of minimum DEPs are depicted in Fig. \ref{fig:DEP_K} where the minimum number of cooperative users, $K_{\min}$ to meet the covert constraint, $\zeta \ge 1 - \epsilon$, for $P_a = 0.67$ and 0.83 are given by 28 and 43, respectively, when $\epsilon = 0.05$. Note that $K_{\min}$ is the value on abscissa corresponding to $\zeta_{\min} = 1 - \epsilon = 0.95 $ on the ordinate. The $K_{\min}$ values obtained from the numerical evaluations are compared with the theoretical result in \eqref{eq:DEP}, which substantiates the theoretical derivation for $K_{\min}$. In addition, the results in Fig. \ref{fig:DEP_K} tell that the minimum DEP, $\zeta_{\min}$ is in proportion to the number of cooperative users, $K$. The minimum number of cooperative users, $K_{\min}$ also grows with the increasing transmit power, $P_a$ which is predicted by the closed-form expression of $K_{\min}$ in \eqref{eq:Opt_K}.

\begin{figure}[htb]
\centering
\includegraphics[width=\figwidth]{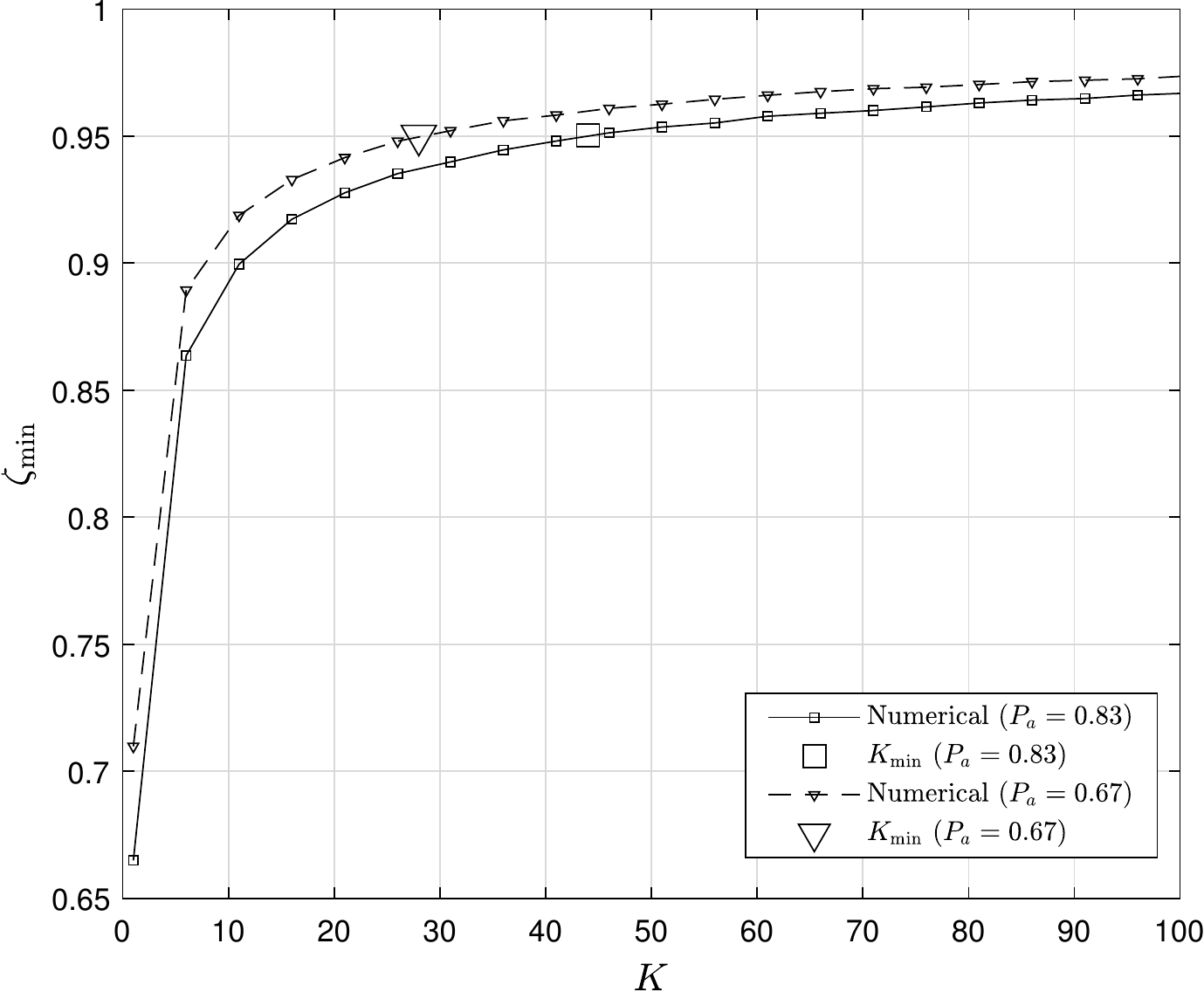}
\caption{Minimum detection error probability, $\zeta_{\min}$, versus the number of cooperative users, $K$, when $\epsilon=0.05$.}
\label{fig:DEP_K} 
\end{figure}

In Fig. \ref{fig:eta_R}, the throughput, $\eta = R P_c$ versus the covert rate, $R$ is depicted for $\epsilon = 0.030$ and 0.025 values when $M=500$, and $P_a = 0.83$. The numerical results in Fig. \ref{fig:eta_R} are simply given by the product of the empirical connection probability and the covert rate. Meanwhile, the theoretical results for the throughput can be readily obtained by evaluating $\eta = R P_c$ where $P_c$ is given by Lemma \ref{lemma:connection_prob}. The empirical and theoretical results in Fig. \ref{fig:eta_R} look close to each other, which confirms the derivation for the connection probability, $P_c$. It is observed that the throughput, $\eta$ is maximized at $R_{\max} = 0.0564$ and 0.0279 for $\epsilon = 0.030$ and $0.025$, respectively which are predicted by the result in Theorem \ref{theorem:Opt_R}. As expected, the maximum throughput gets reduced with the stricter covert constraint, i.e. $\epsilon = 0.025$. The numerical evaluations in Fig. \ref{fig:eta_R} confirm our derivation for $R_{\max}$ in Theorem \ref{theorem:Opt_R}.

\begin{figure}[htb]
\centering
\includegraphics[width=\figwidth]{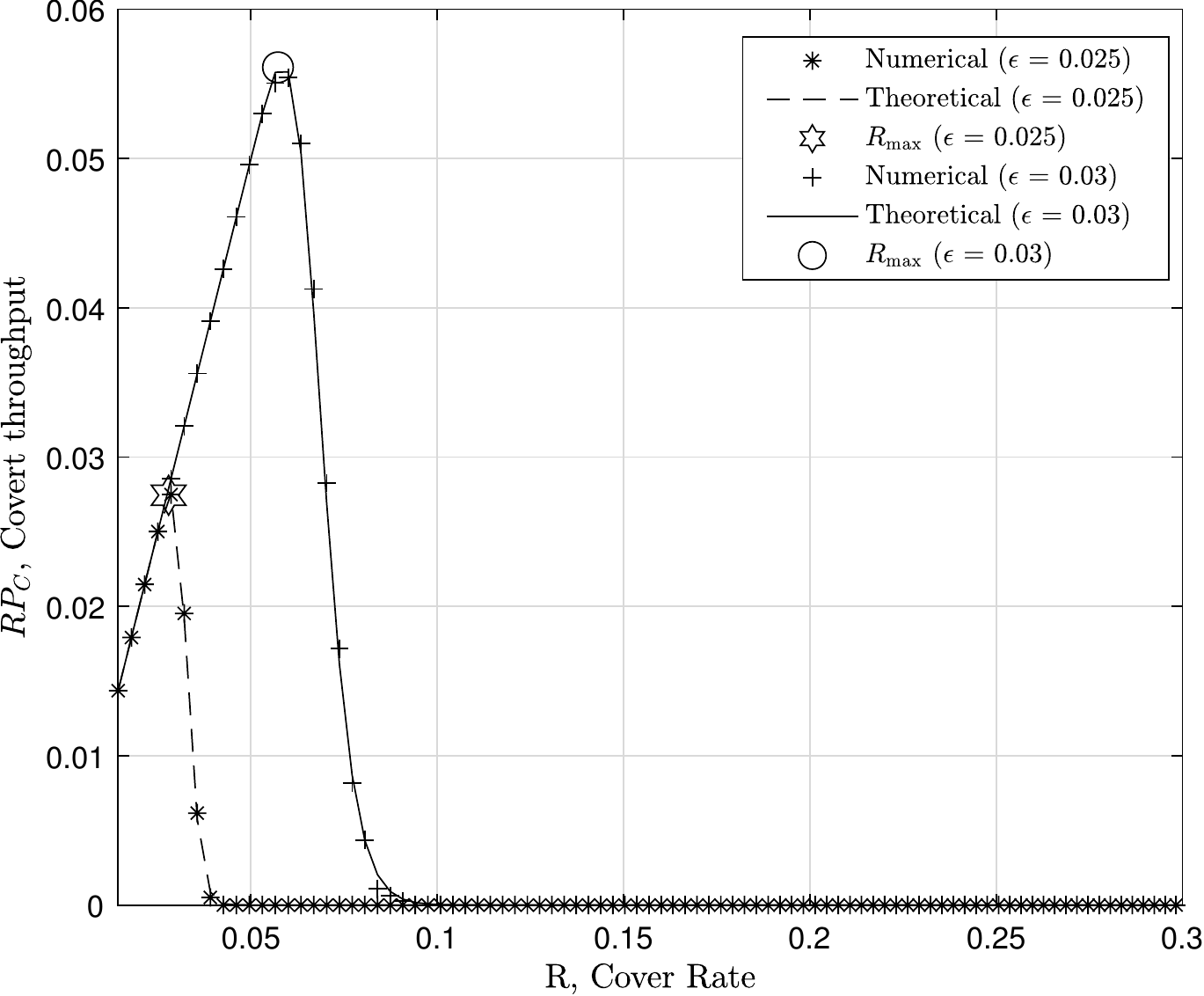}
\caption{Throughput, $\eta$, versus the covert rates, $R$, for $M=500$, $P_a = 0.83$, $\epsilon=0.03$, and $\epsilon=0.025$.}
\label{fig:eta_R} 
\end{figure}

In Fig. \ref{fig:eta_epsilon}, we evaluate the peak maximum throughputs, i.e., $\eta^*_{\max} = \eta_{\max}(P^\ast_a)$ versus the covert constraint, $\zeta \ge 1 - \epsilon$ for $M = 1000$, $M = 500$, and $M = 200$ when $P_{\max} = 1$. The peaks of the maximum throughput are found out by conducting the one-dimensional search in \eqref{eq:Opt_Pa}, which is denoted by `Theoretical' in Fig. \ref{fig:eta_epsilon}. Meanwhile, we also numerically find the peaks of the maximum throughput with realizations of the random variable, $S_h$ in \eqref{eq:re_connection}. It is observed that $\eta^*_{\max}$ decreases as the covert constraint gets tighter since the more cooperative users $K_{\min}^\ast$ have to transmit the interfering signals to meet the more stringent covert constraint. In addition, the results in Fig. \ref{fig:eta_epsilon} also show that $\eta^*_{\max}$ increases with the growing number of users, which is due to the multi-user diversity effect. That is, the $K_{\min}$-th weakest channel gain between the users and Bob becomes reduced with a high probability as the number of users increases, which makes all the channel gains of cooperative users reduced. Thus, the undesirable interference due to the cooperative users at Bob gets diminished. Meanwhile, the number of users does not change the statistics of the channel gains between the cooperative users and Willie.

\begin{figure}[htb]
\centering
\includegraphics[width=\figwidth]{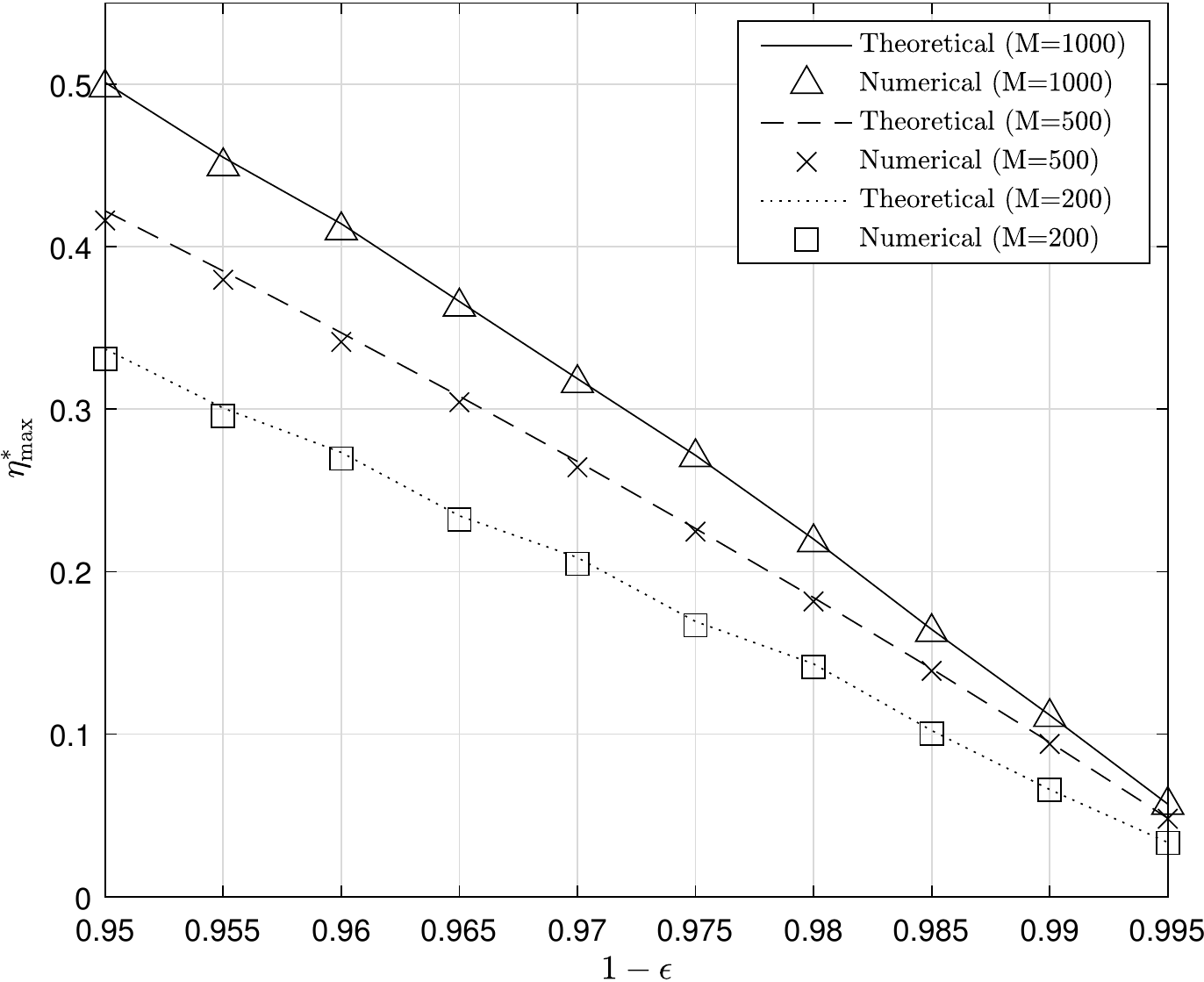}
\caption{The peak maximum throughput, $\eta^*_{\max}$, versus the covert constraint, $\zeta \geq 1 - \epsilon$, for $M=1000$, $M=500$, and $M=200$ when $P_{\max} = 1$.}
\label{fig:eta_epsilon} 
\end{figure}

In Fig. \ref{fig:Pa_K_epsilon_Throughput}, we depict the optimal transmit power of Alice, $P_a^*$ and the optimal number of cooperative users, $K^*_{\min} = K_{\min}(P^\ast_a)$ at $P_a = P^\ast_a$ versus the covert constraint, $\zeta \ge 1 - \epsilon$ for $M = 500$. It is noticed that $P_a^*$ tends to linearly increase as the covert constraint, $1 - \epsilon$, decreases from unity. Note that the decrease of the covert constraint makes it possible to increase the transmit power, $P_a$ and/or reduce the minimum number of cooperative users at $P_a = P^\ast_a$, i.e., $K^\ast_{\min}$. However, it is interesting to see that only the transmit power increases in Fig. \ref{fig:Pa_K_epsilon_Throughput} until $P^\ast_a$ reaches $P_{\max}$ while the minimum number of cooperative users, $K^\ast_{\min}$ stays the same. In Appendix \ref{Appendix:CrossPoint}, we show that $P^\ast_a$ linearly grows with the increasing $\epsilon$ until the transmit power, $P^\ast_a$ reaches its maximum value. Afterward, the decrease of $K^\ast_{\min}$ happens. In Appendix. \ref{Appendix:CrossPoint}, we also investigate the cross point of the minimum DEP, denoted by $\xi$ in Fig. \ref{fig:Pa_K_epsilon_Throughput}, at which the transmit power $P^\ast_a$ reaches its maximum value.  In \eqref{eq:cross_point}, $\xi$ is given by
\[
  \xi = \frac{-\rho+\sqrt{\rho^2 + 16}}{8}
\]
where 
\[ 
  \rho = \sqrt{\frac{\pi}{3} \left(\sqrt{\frac{24M\sigma_b^2}{P_{\max}}+1}-1\right)}.
\]

\begin{figure}[btb]
\centering
\includegraphics[width=\figwidth]{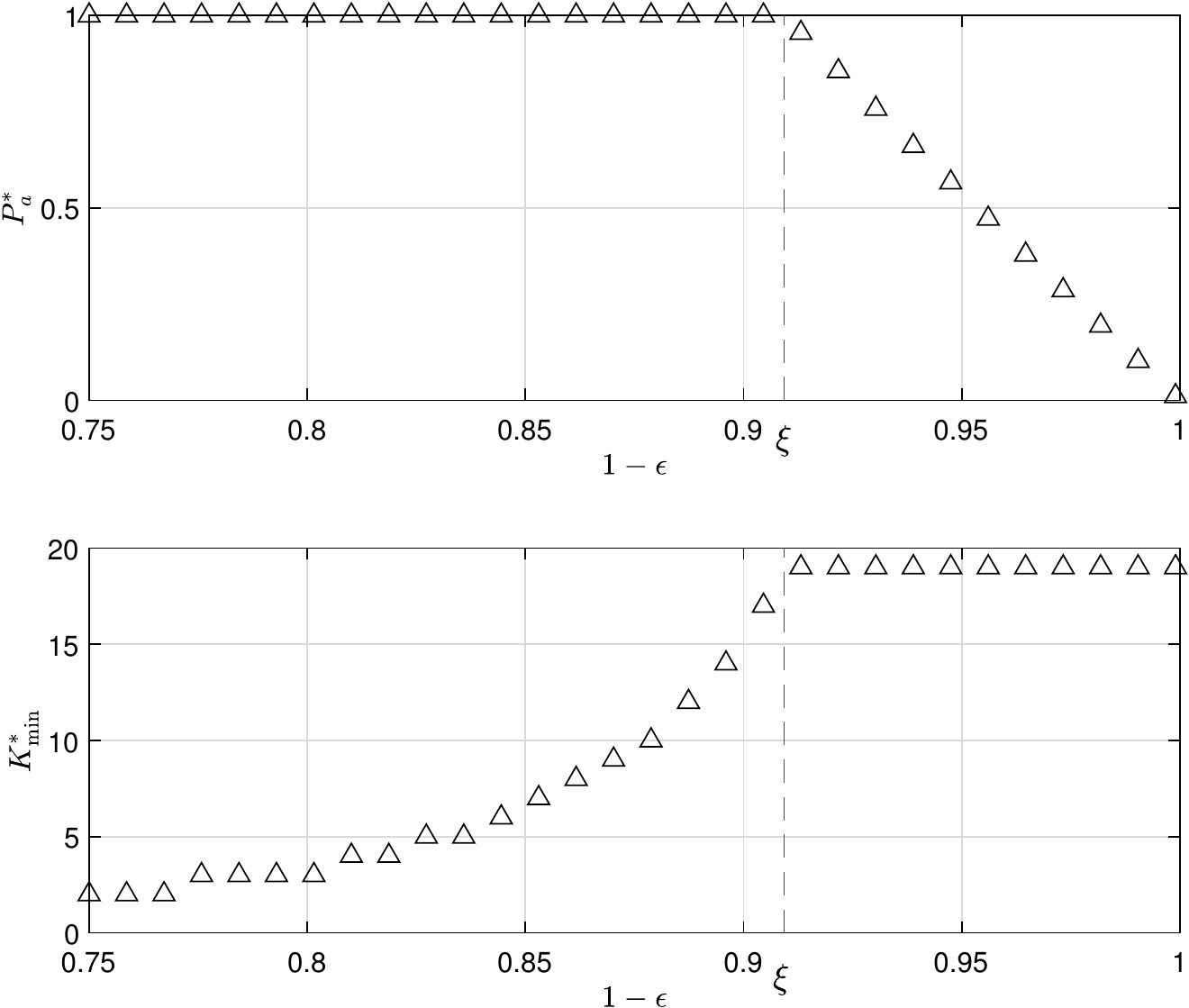}
\caption{Optimal transmit power of Alice, $P_a^*$, and the optimal number of cooperative users, $K^*_{\min}$, versus the covert constraint, $\zeta \ge 1 - \epsilon$ when $M=500$.}
\label{fig:Pa_K_epsilon_Throughput} 
\end{figure}

Now, we perform the optimization problem in \eqref{eq:Opt_EE_Pa} where the energy efficiency, i.e., $E_{\rm ff}$ in \eqref{eq:Eff}, is maximized while satisfying the covert constraint. In Fig. \ref{fig:eta_EE_epsilon}, we compare the throughputs, $\eta^\ast_{\max}$, for the two optimization problems, i.e., the ones with/without the consideration of energy efficiency, versus the covert constraint, $1 - \epsilon$ for $M = 500$. It can be observed that both the throughputs increase with the decreasing covert constraint. The energy efficiencies in \eqref{eq:Opt_EE_Pa} are also compared in Fig. \ref{fig:eta_EE_epsilon} where the design with the optimization problem in \eqref{eq:Opt_EE_Pa} has much higher energy efficiencies throughout all the covert constraints.

\begin{figure}[htb]
\centering
\includegraphics[width=\figwidth]{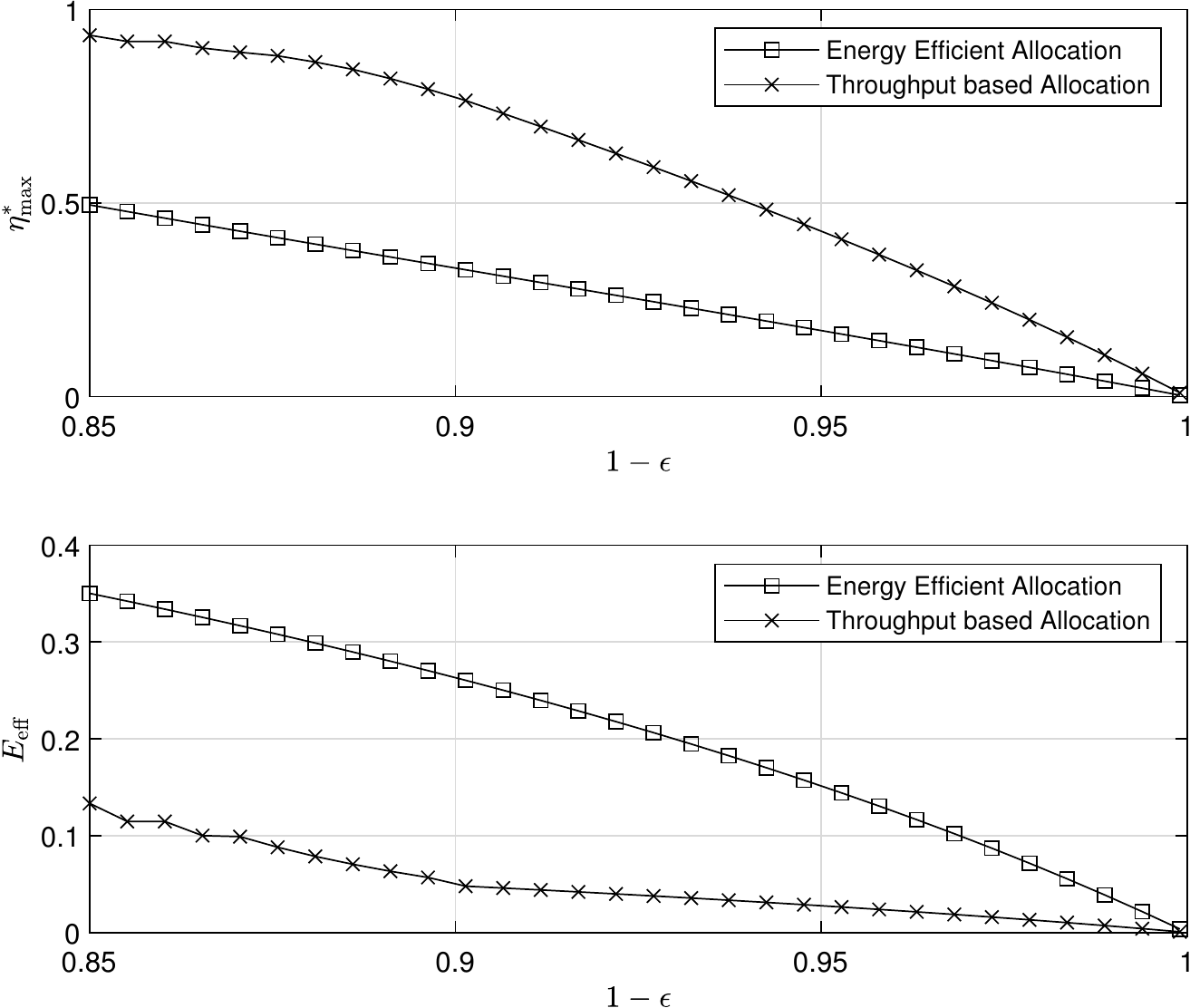}
\caption{Optimal throughput, $\eta^*$, and energy efficiency, $E_{\rm eff}$, for optimal $P_a^*$ versus $1-\epsilon$ according to throughput based allocation and energy efficient allocation when $M=500$.}
\label{fig:eta_EE_epsilon} 
\end{figure}

In Fig. \ref{fig:Pa_K_TotalPower_epsilon}, the two system designs are compared in terms of $P^\ast_a$ and $K^\ast_{\min}$ versus $1-\epsilon$ for $M=500$. In both the designs, the transmit power $P^\ast_a$ linearly increases as the covert constraint, $1 - \epsilon$ decreases. Note that the energy-efficient design has only one cooperative user to minimize power consumption. Since the power of the interference signals from the cooperative users is minimized, the transmit power, $P^\ast_a$ is proportionally reduced as compared to the design without the consideration of energy efficiency. As expected, it is noticed that the total power, i.e., $K_{\min}^*P_{\max} + P_a^*$ has a big discrepancy between the two designs.

\begin{figure}[htb]
\centering
\includegraphics[width=\widefigwidth]{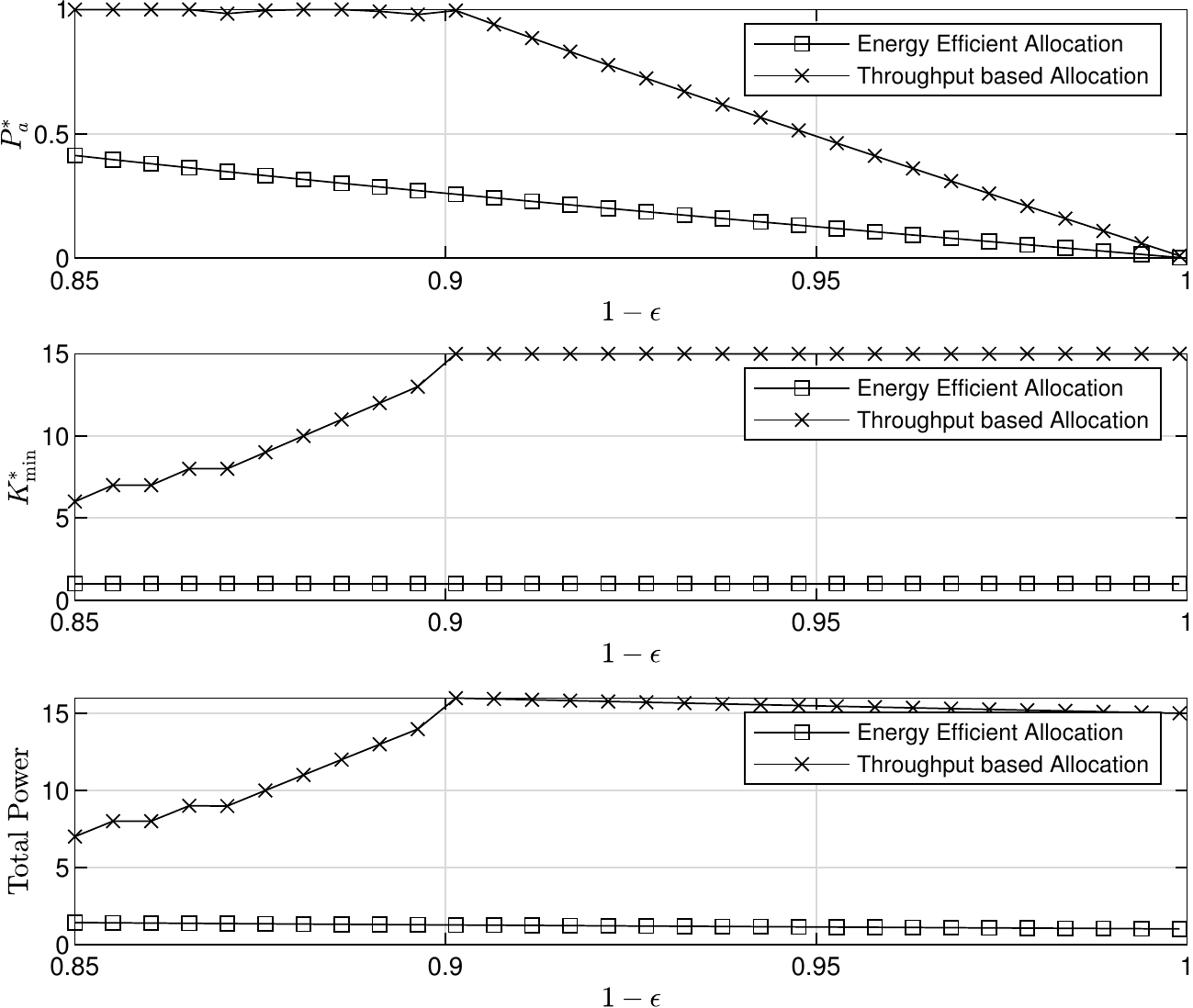}
\caption{Optimal transmit power of Alice, $P_a^*$, the optimal number of cooperative users, $K_{\min}^*$, and total power, $K_{\min}^*P_{\max} + P_a^\ast$, versus $1-\epsilon$ for throughput based allocation and energy efficient allocation when $M=500$.}
\label{fig:Pa_K_TotalPower_epsilon} 
\end{figure}

\section{Conclusions}\label{Sec:Conclusions}
In this paper, we investigate an uplink covert communication scheme with cooperative users which takes advantage of channel independence and multi-user diversity. For the covert communication scheme, we derive important system parameters such as the optimal power profile of cooperative users, the minimum number of cooperative users, and the optimal activation threshold in closed forms. We also derive a closed-form expression of the covert rate with which an analytic expression of the maximum throughput is obtained as a function of the transmit power $P_a$. The work greatly simplifies the multi-dimensional optimization problem to a simple one-dimensional search of the transmit power. The analytic results are confirmed with simulation results under various situations. In addition, the analytic results are further extended to tell interesting behaviors of the proposed scheme such as the changes of $P^\ast_a$ and $K^\ast$ as the covert constraint varies. Finally, we show that the results of work can be utilized in versatile ways by demonstrating a design of covert communication with energy efficiency into account.

\appendices

\section{Proof of Proposition \ref{prop:power_profile}} \label{Appendix:power_profile}
For given $R$ and $P_a$, the optimization problem of (\ref{eq:original_opt}) can be expressed as
\begin{subequations}\label{eq:powerprofile_Opt}
\begin{align}
  \bar{P}^* & = \argmax_{\bar{P}} P_c(\bar{P}) \nonumber\\
  & = \argmax_{\bar{P}} \frac{P_a |h^a_{a,b}|^2}{\sum_{m=1}^M P_m|h^a_{m,b}|^2 + \sigma_b^2}\\
  & s.t. ~~~ \zeta \geq 1-\epsilon\\
  & ~~~~~ 0 \leq P_m \leq P_{\max}.
\end{align}
\end{subequations}
Using the Lyapunov central limit theorem, we can respectively express the probability density functions (PDFs) of the test statistic in \eqref{eq:test_statistic}, denoted by $p\left(T_w ; H_0\right)$ and $p\left(T_w ; H_1\right)$ as
\begin{align*}
   p\left(T_w ; H_0\right) &= \mathcal{N}(\omega,\Omega) \\ 
   p\left(T_w ; H_1\right) &= \mathcal{N}(\omega+P_a,\Omega+P_a^2)
 \end{align*} 
 where $\omega = \sum_{m=1}^M P_m$ and $\Omega = \sum_{m=1}^M P_m^2$. In addition, from (\ref{eq:DEP}), the detection error probability is derived as
 \begin{align*}
    \zeta  & = P_{\rm FA} + P_{\rm MD}\\
    &= \Pr(T_w > \gamma ; H_0) + \Pr(T_w \leq \gamma ; H_1)\\
    & = Q\left(\frac{\gamma-\omega}{\sqrt{\Omega}}\right) + Q\left(\frac{\omega + P_a - \gamma}{\sqrt{\Omega+P_a^2}}\right).
 \end{align*}
 For $P_{\rm FA} = \alpha$,
 \[
\gamma = \sqrt{\Omega}Q^{-1}(\alpha) + \omega,
 \]
 so that
 \[
P_{\rm MD} = Q\left(\frac{P_a - \sqrt{\Omega}Q^{-1}(\alpha)}{\sqrt{\Omega+P_a^2}}\right).
 \]
Finally
\begin{equation}\label{eq:DEP_Pi}
\zeta = \alpha + Q\left(\frac{P_a - \sqrt{\Omega}Q^{-1}(\alpha)}{\sqrt{\Omega+P_a^2}}\right).
\end{equation}
Note that $\zeta$ is a monotonically increasing function of $\Omega$ for any $0 \le \alpha \le 1$.

It will be shown that the minimum of $\zeta$ is also a monotonically increasing function of $\Omega$. For given $\Omega$ and $\hat\Omega$, there exist $\alpha^\ast$ and $\hat\alpha$ such that $\zeta(\alpha, \Omega)$ and $\zeta(\alpha, \hat\Omega)$ are minimized at $\alpha^*$ and $\hat\alpha$, respectively. Without loss of generality, we assume that $\Omega < \hat\Omega$. Since $\zeta(\alpha^\ast, \Omega) \le \zeta(\alpha, \Omega)$ for all $\alpha$, we have 
\[
  \zeta_{\min}(\Omega) = \zeta(\alpha^*,\Omega) \leq \zeta(\hat{\alpha},\Omega).
\]
Thus, 
\[
  \zeta_{\min}(\Omega) = \zeta(\alpha^*,\Omega) \mathop{\leq}_{(a)} \zeta(\hat{\alpha},\Omega) \mathop{<}_{(b)} \zeta_{\min}(\hat \Omega) = \zeta(\hat{\alpha},\hat{\Omega}),
\]
where (a) is due to the minimality of $\alpha^*$ for $\Omega$, and (b) is from the fact that $\zeta$ is an increasing function of $\Omega$ for a fixed $\hat{\alpha}$. The inequalities tell that $\zeta_{\min}(\Omega)$ is a monotonically increasing function of $\Omega$. Thus, the covert constraint, $\zeta_{\min}(\Omega) \ge 1 - \epsilon$ is turned into a lower bound on the total interference power
\begin{equation}\label{eq:convert_DEP_Pi}
  \Omega \geq \delta,
\end{equation}
where $\delta = \zeta^{-1}_{\min}(1-\epsilon)$. By using the Cauchy-Schwarz inequality \cite{Steele04Cauchy}, we have 
\begin{equation} \label{eq:Cauchy}
  \sum_{m=1}^M \left(\frac{P_m}{M}\right)^2 \mathop{\ge}_{(a)} \frac{1}{M}\left(\sum_{m=1}^M \frac{P_m}{M}\right)^2 \ge \frac{\delta}{M^2}.
\end{equation} 
Thus, it is clear that a power profile $\bar P$ satisfying $\sum_{m=1}^M P_m \ge \tilde\delta$  for $\tilde\delta = \sqrt{M\delta}$ also satisfies the inequality in \eqref{eq:convert_DEP_Pi}. It is also well known that for sufficiently large $M$, the inequality in \eqref{eq:Cauchy} denoted by (a) in \eqref{eq:Cauchy} gets tight. Thus, in this work, we replace the constraint in \eqref{eq:convert_DEP_Pi} with 
\[
  \sum_{m=1}^M P_m \ge \tilde\delta.
\]
Then, the optimization problem of (\ref{eq:powerprofile_Opt}) is reformulated as
\begin{subequations}\label{eq:powerprofile_Opt_reformul}
\begin{align}
  \min_{\bar{P}}& \sum_{m=1}^M P_m|h^a_{m,b}|^2 + \sigma_b^2  \label{eq:powerprofile_Opt_reformul_obj}\\
  & s.t. ~~~ \sum_{m=1}^M P_m \ge \tilde\delta, \\
  & ~~~~~ 0 \leq P_m \leq P_{\max}.
\end{align}
\end{subequations}
We find the optimal $\bar{P}^*$ to satisfy the Karush-Kuhn-Tucker (KKT) conditions of (\ref{eq:powerprofile_Opt_reformul}) which are summarized as follows:
\begin{itemize}
  \item
  Stationarity: $|h^a_{m,b}|^2 - \lambda_m^* + \omega_m^* - \nu^*  = 0$, $\forall\,m$,
  \item
  Complementary slackness: $\lambda_m^* P_m^* = 0$, $\omega_m^*(P_m^* - P_{\max})=0$, $\forall\,m$, and $\nu^* (\tilde\delta  - \sum_{m=1}^{M}P_m) = 0$, 
  \item
  Primal feasibility: $P_m^* \geq 0$, $P_m^*\leq P_{\max}$, $\forall\,m$ and $\sum_{m=1}^M {P_m^*} \ge \tilde\delta$, 
  \item
  Dual feasibility: $\lambda_m^*$, $\omega_m^*$, and $\nu^* \geq 0$, $\forall\,m$.
\end{itemize}
From the stationarity, we have $\lambda_m^* = |h^a_{m,b}|^2 + \omega_m^* -\nu^*$, which allows us to express the conditions as
\begin{enumerate}\label{enum:KKT}
  \item 
  $P_m^* \geq 0$, $P_m^* \leq P_{\max}$, $\omega_m^* \geq 0$, and $\nu^* \ge 0$,
  \item
  $\sum_{m=1}^M {P_m^*} \ge \tilde\delta$,
  \item
  $(|h^a_{m,b}|^2 + \omega_m^* - \nu^*)P_m^* = 0$,
  \item
  $\omega_m^* \left(P_m^* - P_{\max}\right) = 0$,
  \item
  $|h_{m,b}|^2 + \omega_m^* -\nu^* \geq 0$,
  \item
  $\nu^*(\tilde\delta  - \sum_{m=1}^{M}P_m) = 0$.
\end{enumerate}

\begin{table}[htb] 
\caption{Solutions with the KKT conditions; NS stands for no solution.}
\label{Tbl:KKT_Conditions}
\centering
\begin{tabular}{|>{\centering\arraybackslash} m{1cm}|>{\centering\arraybackslash} m{2cm}|>{\centering\arraybackslash} m{2cm}|>{\centering\arraybackslash} m{2cm}|}
\toprule 
  & $|h^a_{m,b}|^2 > \nu^*$ &$|h^a_{m,b}|^2 = \nu^*$ & $|h^a_{m,b}|^2  < \nu^*$   \\
\midrule
  $ \omega_m^* = 0 $ & $P_m^* = 0$ due to 3) & $P_m^* = \tilde\delta - \sum_{|h^a_{m,b}|^2 \neq \nu^*}P_m^*$ due to 6)& NS due to 5)\\ \hline
  $ \omega_m^* > 0$ & NS due to 3) and 4)  & NS due to 3) and 4)& $P_m^* = P_{\max}$ due to 4) \\ 
\bottomrule
\end{tabular}
\end{table}   

Based on the KKT conditions, the solution for the optimization problem in \eqref{eq:powerprofile_Opt_reformul} can be found in Table \ref{Tbl:KKT_Conditions} where the rows and columns represent the values of $\omega^\ast_m$ and the relations between $|h^a_{m, b}|^2$ and $\nu^\ast$, respectively. For some combinations, we do not have a solution that is denoted by `NS' in Table \ref{Tbl:KKT_Conditions}. Thus, the optimal $P_m^*$ can be expressed as 
\begin{align}
    P_m^*= 
  \begin{cases}
      P_{\max}, & |h^a_{m,b}|^2 < \tau \\
      0, & |h^a_{m,b}|^2 > \tau\\
      \displaystyle \tilde\delta - \sum_{m:|h^a_{m,b}|^2 < \tau}P_{\max}, & |h^a_{m,b}|^2 = \tau
  \end{cases},
\end{align}
where $\tau = \nu^*$. Note that the solution for $|h^a_{m,b}|^2 = \tau$, i.e., $P^\ast_m =  \tilde\delta - \sum_{m:|h^a_{m,b}|^2 < \tau}P_{\max}$ can be expressed as $P_{\max} = \tilde \delta - |\{m : |h^a_{m, b}|\}| P_{\max}$ which requires the number of users with their channel gains satisfying $|h^a_{m,b}|^2 < \tau$. However, the solution is not feasible since users do not know the channel gains of other users. Thus, for the given system model in Section \ref{Sec:System}, the optimal $P_m^*$ is determined as
\begin{align}
    P_m^*= 
  \begin{cases}
      P_{\max}, & |h^a_{m,b}|^2 \le \tau \\
      0, & |h^a_{m,b}|^2 > \tau. \\
   \end{cases}
\end{align} 
\qed

\section{Proof of Lemma \ref{lemma:DEP}} \label{Appendix:DEP}
 The false-alarm probability at Willie is given by
\begin{align*}
  P_{\text FA} & = \Pr(T_w > \gamma ; H_0) \\
    & = \Pr(T_w - \sigma_w^2 > \gamma - \sigma_w^2 ; H_0) \\
    & \mathop{\approx}_{(a)} \Pr(Z > \gamma-\sigma_w^2)
      =Q\left(\frac{(\gamma-\sigma_w^2) -KP_{\max}  }{\sqrt{KP_{\max}^2}}\right),
\end{align*}
where the approximation in (a) is carried out by applying the central limit theorem, i.e., $\Gamma(k,P_i) \approx \mathcal N(kP_i,kP_i^2)$, and $Z$ is a random variable following $N(K P_{\max},K P_{\max}^2)$ for the number of cooperative users, $K$. Meanwhile, the miss-detection probability at Willie can be derived as
\begin{align*}
  P_{MD} & = \Pr(T_w \leq \gamma ; H_1) = \Pr(T_w-\sigma_w^2 \leq \gamma-\sigma_w^2 ;H_1)  \\
  & \mathop{\approx}_{(b)}   \Pr(X+Y \leq \gamma-\sigma_w^2) \\
  & =  \int_{-\infty}^{\infty} f_{Y}(y) \Pr\left(X\leq \gamma-\sigma_w^2-y\right)dy \\
  & = Q\left(\frac{KP_{\max} - (\gamma-\sigma_w^2) }{\sqrt{KP_{\max}^2}}\right)\\
   &\hspace{0.1\columnwidth} -  \exp\left(K\frac{P_{\max}^2+2P_aP_{\max}}{2P_a^2} - \frac{ \gamma-\sigma_w^2}{P_a}\right) \\
    & \hspace{0.1\columnwidth} \times Q\left(\sqrt{K}\frac{P_{\max}+P_a}{P_a} - \frac{1}{P_{\max}\sqrt{K}} (\gamma-\sigma_w^2)\right),
\end{align*}
where the random variables, $X$ and $Y$ follow $\Gamma(1, P_a)$ and $N(KP_{\max},KP_{\max}^2)$, respectively, and the approximation in (b) is carried out by applying the central limit theorem of $\Gamma(K,P_{\max})$. 
\qed


\section{Proof of Lemma \ref{lemma:connection_prob}} \label{Appendix:connection_prob}
Let $h^a_{m_1,b} \leq h^a_{m_2,b} \leq \ldots \leq h^a_{m_M,b}$ denote the order statistics in a sample of $M$ from the standard exponential distribution, i.e., the rate parameter of unity. Then, $h^a_{m_r, b}$ can be expressed by 
\[
h^a_{m_r, b} = \sum_{i=1}^r \frac{y_i}{M-i+1},
\]
where $y_i$'s are independent and identically distributed (i.i.d) standard exponential random variables \cite{Renyi53OntheTheory}. Then, the sum of order statistics, $S_h = \sum_{i=1}^{K_{\min}} h^a_{m_i, b}$ can be expressed as
\begin{equation} \label{eq:sumOfOSD}
  S_h = \sum_{j = 1}^{K_{\min}} \frac{K_{\min} - j + 1}{M - j + 1} y_j  = \sum_{j = 1}^{K_{\min}} \Delta_j y_j,
\end{equation}
where $\Delta_j \triangleq (K_{\min} - j + 1)/(M - j + 1)$, and each term in the summation follows a Gamma distribution, i.e.,
\[
  \Delta_j y_j \sim \Gamma(1, \Delta_j).
\]
Thus, the PDF for the sum of the order statistics in \eqref{eq:sumOfOSD} can be expressed as the convolution of the Gamma distributions as follows:
\[
  S_h \sim \Conv_{j = 1}^{K_{\min}} \Gamma(1, \Delta_j),
\]
which can be approximated as a Gaussian distribution due to the Lyapunov central limit theorem. That is,
\begin{equation} \label{eq:S_h}
  S_h \sim \mathcal{N}(\mu,\Xi^2),
\end{equation}
where $\mu = \sum_{i=1}^{K_{\min}} \Delta_i$ and $\Xi^2 = \sum_{i=1}^{K_{\min}} \Delta_i^2$. Therefore, the connection probability can be derived as
\[
  P_c  = \Pr(S_h \leq T) = Q\left(\frac{\mu-T}{\Xi}\right),
\]
where $T = \frac{1}{r}\frac{P_a|h^a_{a,b}|^2}{P_{\max}} -\frac{\sigma_b^2}{P_{\max}}$. This completes the proof. 
\qed

\section{Proof of Theorem \ref{theorem:Opt_R}} \label{Appendix:Opt_R}

From (\ref{eq:connection_prob}), the throughput, $\eta = R P_c$ can be written as 
  \begin{equation} \label{eq:thrput2}
    \eta = \log_2\left(1+r\right) Q\left(\frac{\mu P_{\max} + \sigma_b^2- \frac{P_a|h^a_{a,b}|^2}{r} }{\Xi P_{\max}}\right).
  \end{equation}
The relation between $K$ and $\epsilon$ in \eqref{eq:Opt_K} tells $K$ quadratically grows with the decreasing $\epsilon$ value. In the large DEP regime, i.e., $\epsilon \ll 1$, there are a large number of cooperative users, which in turn reduces the value of $r$, and thus we can have $\log(1 + r) \approx r$. In such a case, the throughput in \eqref{eq:thrput2} can be expressed as
  \begin{equation}\label{eq:taylor_throughput}
    \eta = \frac{r}{\log 2} Q\left(\frac{\mu P_{\max} + \sigma_b^2 - \frac{P_a|h^a_{a,b}|^2}{r} }{\Xi P_{\max}}\right)= \frac{r}{\log 2}Q\left(\theta - \frac{\kappa}{r}\right),
  \end{equation}
where  $\theta = (\mu P_{\max}+\sigma_b^2)/(P_{\max} \Xi)$, $\kappa = P_a|h^a_{a,b}|^2/(P_{\max}\Xi)$, and $\Xi$ and $\mu$ are defined in \eqref{eq:S_h}. For the throughput in  (\ref{eq:taylor_throughput}), the derivative of throughput is approximated as
\begin{align}\label{eq:deriv_eta}
 & \frac{d}{dr} \left(\frac{r}{\log 2}Q\left(\theta-\frac{\kappa}{r}\right)\right) \nonumber\\
 & = \frac{r}{\log2} \left(-\frac{1}{\sqrt{2\pi}}\exp{-\frac{\left(\theta-\frac{\kappa}{r}\right)^2}{2}}\frac{\kappa}{r^2} \right) + \frac{1}{\log 2}Q\left(\theta - \frac{\kappa}{r}\right) \nonumber\\
 & \approx -\frac{\theta}{\sqrt{2\pi}\log 2}\exp{-\frac{\left(\theta-\frac{\kappa}{r}\right)^2}{2}} + \frac{1}{2 \log 2},
\end{align} 
where using some results in \cite{Decker75Computer}, the second term is approximated as 
\begin{align} \label{eq:erfc_approx}
  & \frac{1}{\log 2}Q\left(\theta - \frac{\kappa}{r}\right) = \frac{1}{2\log 2}{\rm erfc} \left(\frac{\theta - \frac{\kappa}{r}}{\sqrt 2}\right) \nonumber\\ 
  & \approx \frac{1}{2\log 2}\left(1- \frac{2}{\sqrt \pi} \exp{-\frac{\left(\theta - \frac{\kappa}{r}\right)^2}{2}}\frac{\theta - \frac{\kappa}{r}}{\sqrt 2}\right) \nonumber\\ 
  & = \frac{1}{2\log 2} - \frac{1}{\sqrt{2\pi}\log 2}\left(\theta - \frac{\kappa}{r}\right) \exp{-\frac{\left(\theta - \frac{\kappa}{r}\right)^2}{2}}.\nonumber
\end{align}

We need to find the value of $r$ making the derivative of throughput equal to zero, which provides the covert rate, $R$, maximizing the throughput. It is readily shown that the derivative of throughput will be zero when $r$ is a root of the following quadratic equation of $r$:
\[
  \left(\theta^2 + \psi\right)r^2 - 2\theta\kappa r + \kappa^2 = 0, 
\]
where $\psi = 2\log(\sqrt{2\pi}/2\theta)$. The roots of quadratic equation are 
\[
  r^{\pm} = \frac{2\theta\kappa \pm \sqrt{4\theta^2\kappa^2 - 4\kappa^2\left(\theta^2 + \psi\right)}}{2\left(\theta^2 + \psi\right)} = \frac{\kappa}{\theta \mp \sqrt{-\psi}},
\]
where $r^+ > r^-$. Among the two roots, it will be that $r^{-1}$ is the solution that we look for, i.e., $R_{\max} = \log_2(1 + r^-)$.

For $r \in (r^-, r^+)$, we can express $r$ as 
\begin{equation} \label{eq:r}
  r  = \frac{\kappa}{\theta + \sqrt{-\psi} - c},
\end{equation}
where $c \in (0, 2\sqrt{-\psi})$. Then, the derivative of $\eta$ with respect to $r$ can be expressed in terms of $c$ as
\begin{align*}
  \frac{d \eta}{dr} & = \frac{1}{2 \log 2}\left(1 - \frac{2\theta}{\sqrt{2\pi}}\exp{-\frac{\left(- \sqrt{-\psi} + c \right)^2}{2}} \right)\\
  & = \frac{1}{2 \log 2}\left(1 - \frac{2\theta}{\sqrt{2\pi}}\exp{\frac{\psi - c^2 + 2c\sqrt{-\psi}}{2}}\right)\\
  & = \frac{1}{2 \log 2}\left(1 - \frac{2\theta}{\sqrt{2\pi}}\exp{\frac{\psi}{2}}\exp{\frac{c(2\sqrt{-\psi}-c)}{2}}\right)\\
  & \mathop{=}_{(a)} \frac{1}{2\log{2}}\left(1-\exp{\frac{c(2\sqrt{-\psi}-c)}{2}}\right) < 0,
\end{align*}
where the step (a) is due to the fact that $\exp(\psi/2) = (2\theta/\sqrt{2\pi})^{-1}$. It is obvious that the derivative of $\eta$ is negative for $r \in (r^-, r^+)$, which implies $\eta(r^-) > \eta(r^+)$. Thus, we have $r^-$ as the solution with which the maximum covert rate, $R_{\max}$ is obtained as $R_{\max} = \log_2 (1+ r^-) = \log_2 (1+\kappa/(\theta + \sqrt{-\psi}))$. \qed

\section{Proof of Cross Point} \label{Appendix:CrossPoint}
For finding the cross point of covert constraint, we first simplify the maximum throughput in Corollary \ref{corollary:Max_Throughput} as follows:
\begin{align*}
  & \eta_{\max} =\log_2\left(1+\frac{\kappa}{\theta + \sqrt{-\psi}}\right) Q(-\sqrt{-\psi}) \\
  & \mathop{\approx}_{(a)} \frac{\kappa}{\theta + \sqrt{-\psi}}\frac{1}{\log 2}Q(-\sqrt{-\psi}) \\
  & = \frac{1}{\log 2}\frac{\kappa}{\theta + \sqrt{-\psi}}\left(\frac{1}{2} {\rm erfc}\left(-\sqrt{-\log \frac{\sqrt{2\pi}}{2\theta}}\right)\right)\\
  &  \mathop{\approx}_{(b)} \frac{1}{\log 2}\frac{\kappa}{\theta + \sqrt{-\psi}} \\
  & ~~~~~~~~~~~~ \times\left(\frac{1}{2}\left(1+\frac{2}{\sqrt{\pi}}\sqrt{\log \frac{2\theta}{\sqrt{2\pi}}} \exp{-\log \frac{2\theta}{2\pi}}\right)\right) \\
  & =  \frac{1}{\log 2}\frac{\kappa}{\theta + \sqrt{-\psi}}\left(\frac{1}{2} \left(1+\frac{\sqrt{2}}{\theta}\sqrt{\log \frac{2\theta}{\sqrt{2\pi}}}\right) \right)\\
  & =  \frac{1}{\log 2}\frac{\kappa}{\theta + \sqrt{-\psi}} \left(\frac{1}{2\theta}\left(\theta + \sqrt{-\psi}\right)\right) = \frac{\kappa}{2\theta\log 2},
\end{align*}
where (a) comes from the first order Taylor approximation since $r_{\max} = \kappa/(\theta + \sqrt{-\psi}) \ll 1$, and (b) is due to the approximation of the complementary error function in \cite{Decker75Computer,Ren2007Closed}. Then, the maximum throughput can be succinctly expressed as
\begin{equation} \label{eq:simplified_throughput}
  \eta_{\max} = \frac{1}{2\log 2} \frac{P_a|h^a_{a,b}|^2}{P_{\max}\mu +\sigma_b^2},
\end{equation}
where  $\theta = (\mu P_{\max}+\sigma_b^2)/(P_{\max} \Xi)$, $\kappa = P_a|h^a_{a,b}|^2/(P_{\max}\Xi)$, and $\mu$ is defined in \eqref{eq:S_h}. 

In (\ref{eq:simplified_throughput}), the term, $\mu = \sum_{i=1}^{K_{\min}} \frac{K_{\min}+1-i}{M+1-i}$ can be re-expressed as
\begin{align*}
  \mu & = \sum_{i=1}^{K_{\min}} \frac{K_{\min}+1-i}{M+1-i} \\
      & = (K_{\min}-M)(H_M - H_{M-K_{\min}}) + K_{\min},
\end{align*}
where $H_n$ is the harmonic number defined as $H_n=\sum_{i=1}^n \frac{1}{i} = \psi_0(n) + 1/n +\gamma_0$ where $\gamma_0$ and $\psi_0$ are the Euler-Mascheroni constant and the digamma function. Then, we can approximate $\mu$ as 
\begin{align}
  \mu & = (K_{\min}-M)(H_M - H_{M-K_{\min}}) + K_{\min} \nonumber\\
  & = (K_{\min}-M)\left(\psi_0(M) - \psi_0(M-K_{\min}) \right) + K_{\min}\frac{M + 1}{M} \nonumber \\ 
  & \mathop{\approx}_{(c)} (M-K_{\min})\log \frac{M-K_{\min}}{M} + K_{\min}\frac{2M + 1}{2M} \nonumber \\ 
  & = M\log\left(\frac{M-K_{\min}}{M}\right)^{\frac{M-K_{\min}}{M}} + K_{\min}\frac{2M + 1}{2M} \nonumber \\
  & \mathop{\approx}_{(d)} M \left(\frac{M-K_{\min}}{M} - 1 \right) \sqrt{1-\frac{K_{\min}}{M}} + K_{\min}\frac{2M + 1}{2M} \nonumber \\ 
   & \mathop{\approx}_{(e)} -K_{\min}\left(1-\frac{K_{\min}}{2M}\right) + K_{\min}\frac{2M + 1}{2M} \nonumber \\ 
   &  = \frac{K_{\min}(1 + K_{\min})}{2M}. \label{eq:muApprox}
\end{align}
where (c) is due to $\psi_0(x) \approx \log x - \frac{1}{2x}$ \cite{Abramowitz64Handbook}, (d) is possible with $\log x^x \approx (x-1)\sqrt{x}$ for $x \approxeq 1$, and (e) is the first order Taylor approximation for $\sqrt{1-x}$, i.e., $\sqrt{1-x} \approx 1-\frac{x}{2}$. The optimization of the throughput in \eqref{eq:Opt_Pa} is now reformulated with \eqref{eq:simplified_throughput} and \eqref{eq:muApprox} as follows:
\begin{align*}
  \max_{P_a} & \frac{1}{\log 2}\frac{P_a|h^a_{a,b}|^2}{P_{\max} K_{\min}(1 + K_{\min})/(2M) + \sigma_b^2}, \\
  & s.t.~~ P_a \leq P_{\max},
\end{align*}
which is equivalent to 
\begin{align*}
  \min_{P_a} & \frac{1}{P_a|h^a_{a,b}|^2} \left(\frac{P_{\max}K_{\min}(1 + K_{\min})}{2M} + \sigma_b^2\right), \\
  & s.t. ~~ P_a \leq P_{\max}.
\end{align*}
The objective function is expressed in terms of $P_a$ with $K_{\min} = P_a^2 c_{\epsilon}/P_{\max}^2$ where the ceil function in \eqref{eq:Opt_K} is omitted. Then, the objective function becomes a cubic function of $P_a$ as follows:
\[
 aP_a^3 + bP_a + \frac{c}{P_a},
\]
where $a = \frac{c_{\epsilon}^2}{2P_{\max}^3|h^a_{a,b}|^2 M}$, $b = \frac{c_{\epsilon}}{2P_{\max}|h^a_{a,b}|^2 M}$, and $c = \frac{\sigma_b^2}{|h^a_{a,b}|^2}$. It can be readily shown that the roots of the derivative of the objective function are obtained by finding the solutions of the equality, $3aP_a^4 + bP_a^2 - c = 0$. There are four solutions, 
\[
  P_a = \pm \sqrt{\frac{\pm \sqrt{12ac+b^2}-b}{6a}}
\]
among which we have only one valid solution,
\begin{align}\label{eq:approximated_Pa1}
  P_a^* &= \sqrt{\frac{\sqrt{12ac+b^2}-b}{6a}} = P_{\max}\sqrt{\frac{\sqrt{\frac{24M\sigma_b^2}{P_{\max}} + 1 } - 1} {6c_\epsilon}} \\
  & = P_{\max} \sqrt{\frac{\Upsilon}{c_\epsilon}}, \label{eq:approximated_Pa2}
\end{align}
where $\Upsilon = (\sqrt{24M\sigma_b^2/P_{\max} + 1} - 1)/6$ since the transmit power $P_a$ is a non-negative real number. Thus, the search for $P^\ast_a$ has a unique solution.
Note that $P_a \le P_{\max}$. Thus, the optimal $P_a^*$ can be expressed by 
\begin{align}\label{eq:approximated_optimal_Pa}
   P_a^* = 
   \begin{cases}
      P_{\max} \sqrt{\frac{\Upsilon}{c_\epsilon}}, & \Upsilon \leq c_\epsilon \\
     P_{\max}, & {\rm otherwise}
   \end{cases}.
 \end{align} 
From \eqref{eq:approximated_optimal_Pa}, the optimal $K^*$ is given by
\begin{align}\label{eq:approximated_K}
    K^*_{\min} = 
   \begin{cases}
      \lceil \Upsilon \rceil, & \Upsilon \leq c_\epsilon \\
     \lceil c_\epsilon \rceil, & {\rm otherwise}
   \end{cases}.
\end{align}
Remember that as the covert constraint $1 - \epsilon$ decreases from unity, the optimal transmit power $P^\ast_a$ grows until it reaches the maximum $P_{\max}$. In the large DEP regime, i.e., $\epsilon \ll 1$, $c_\epsilon \approx \epsilon^{-2}$, which tells that $P^\ast_a$ linearly grows with the increasing $\epsilon$ value. Now, we get the cross point of the covert constraint below which the transmit power remains at the maximum power by finding the $\epsilon$ value satisfying the following equality:
\[
  \sqrt{\frac{\sqrt{\frac{24M\sigma_b^2}{P_{\max}} + 1 } - 1} {6c_\epsilon}} = 1.
\]
The solution of the equality is readily obtained as
\[
  \epsilon = \frac{-\rho + \sqrt{\rho^2 + 16}}{8},
\]
and thus we have the cross point of the covert constraint as
\begin{equation}\label{eq:cross_point}
  \xi = 1 - \frac{-\rho + \sqrt{\rho^2 + 16}}{8},
\end{equation}
where $\rho = \sqrt{2\pi \Upsilon}$.
\qed

\bibliographystyle{IEEEtran}
\bibliography{ref_covert}
\end{document}